\documentclass[11pt]{amsart}	
\usepackage{hyperref}
\hypersetup{
    colorlinks=true,
    linkcolor=blue,
    filecolor=magenta,      
    urlcolor=cyan,
}
\usepackage{graphicx}
\usepackage[utf8x]{inputenc}
\usepackage[T1]{fontenc}

\usepackage{url}

\usepackage{proof}
\usepackage{amssymb}
\usepackage{amsmath,amsfonts}

\usepackage{mathtools}
\usepackage{bbm}
\usepackage{wrapfig}

\usepackage{array}

\usepackage{mathrsfs}

\usepackage{tikz}
\usetikzlibrary{shapes, arrows, decorations.markings, calc}

\usepackage{algorithm}
\usepackage[noend]{algpseudocode}

\usepackage{array, multirow}

\usepackage{ifthen}
\usepackage{mdframed}

\newboolean{showComments}
\setboolean{showComments}{false}

\newcommand{\comment}[3]{\ifthenelse{\boolean{showComments}}{%
\begin{mdframed}
\noindent\textbf{Reviewer comment:} {(#1)}
\begin{quotation}
{#2}
\end{quotation}

\hrule

\vspace{8pt}

\noindent\textbf{Authors comment:}

{#3}
\end{mdframed}
}%
{}}

\newcommand{\rulelbl}[1]{\small{\textrm{(#1)}}}
\newcommand{\fun}[1]{\small{\textsf{#1}}}

\newcommand{\PO}{{\mathscr{P}\!\mathscr{O}}}
\newcommand{\LE}{{\mathscr{L}\!\mathscr{E}}}

\tikzset{ordedge/.style={draw,decoration={
    markings,
    mark=at position 0.5 with {\arrow[very thick]{>}}},
    postaction={decorate}
    }}

\tikzset{my_ordedge/.style={draw,decoration={
    markings,
    mark=at position 0.4 with {\arrow[very thick]{>}}},
    postaction={decorate}
    }}

\newcolumntype{M}{>{\centering\arraybackslash}m{\dimexpr.21\linewidth-2\tabcolsep}}
\newcolumntype{L}{>{\centering\arraybackslash}m{\dimexpr.4\linewidth-2\tabcolsep}}
\newcolumntype{H}{>{\centering\arraybackslash}m{\dimexpr.13\linewidth-2\tabcolsep}}
\newcolumntype{S}{>{\centering\arraybackslash}m{\dimexpr.1\linewidth-2\tabcolsep}}
\newcolumntype{A}{>{\centering\arraybackslash}m{\dimexpr.08\linewidth-2\tabcolsep}}

\newcommand{\dd}{\mathrm{d}}
\newcommand{\Unif}{\textsc{Uniform}}

\newcommand{\entails}[1]{\vdash_{\scriptscriptstyle {#1}}}

\newtheorem{theorem}{Theorem}

\newtheorem{corollary}[theorem]{Corollary}

\theoremstyle{definition}

\newtheorem{definition}[theorem]{Definition}

\newtheorem{example}[theorem]{Example}

\begin{document}
\title[The Combinatorics of Barrier Synchronization]{The Combinatorics of Barrier Synchronization$^\dag$\footnote{
       $^\dag$ This research was partially supported by the ANR MetACOnc project ANR-15-CE40-0014.}}

\author{Olivier Bodini}
\address{Laboratoire d'Informatique de Paris-Nord,
    CNRS UMR 7030 - Institut Galil\'ee - Universit\'e Paris-Nord,
    99, avenue Jean-Baptiste Cl\'ement, 93430 Villetaneuse, France.} 
\email{Olivier.Bodini@lipn.univ-paris13.fr}

\author{Matthieu Dien}
\address{Université de Caen -- GREYC -- CNRS UMR 6072.} 
\email{Matthieu.Dien@unicaen.fr}

\author{Antoine Genitrini}
\address{Sorbonne Universit\'e,	CNRS,
    Laboratoire d'Informatique de Paris 6 -LIP6- UMR 7606, F-75005 Paris, France.}
\email{Antoine.Genitrini@lip6.fr}

\author{Fr\'ed\'eric Peschanski}
\address{Sorbonne Universit\'e,	CNRS,
    Laboratoire d'Informatique de Paris 6 -LIP6- UMR 7606, F-75005 Paris, France.}
\email{Frederic.Peschanski@lip6.fr}

\maketitle              
\begin{abstract}
    In this paper we study the notion of synchronization from the point of
    view of combinatorics. As a first step, we address the quantitative
    problem of counting the number of executions of simple processes
    interacting with synchronization barriers. We
    elaborate a systematic decomposition of processes that produces a
    symbolic integral formula to solve the problem. Based on this
    procedure, we develop a generic algorithm to generate process
    executions uniformly at random. For some interesting sub-classes of processes
    we propose very efficient counting and random sampling algorithms.
    All these algorithms have one important characteristic in common: they work on the
    control graph of processes and thus do not require the explicit
    construction of the state-space.
    
\end{abstract}

\comment{Review 3}{Perhaps the authors have been too ambitious in their aim to include counting, subclasses, and an experimental study, all in one conference contribution together with (sketches of) proofs.}%
{We agree that the breadth of the paper is rather large but we think it is still a whole contribution and decided not to remove a full section. However we were able to save some space in less important parts. Hence we added further explanations, especially a detailed description of Algorithm 1 (cf. the correponding comment box).}

\section{Introduction}

The objective of our (rather long-term) research project is to study
the \emph{combinatorics} of concurrent processes.  Because the
mathematical toolbox of combinatorics imposes strong constraints on
what can be modeled, we study \emph{process calculi} with a very restricted
focus. For example in ~\cite{BGP16} the processes we study can only perform
atomic actions and fork child processes, and in~\cite{BGP13} we enrich
this primitive language with \emph{non-determinism}.  In the present
paper, our objective is to isolate another fundamental ``feature'' of
concurrent processes: \emph{synchronization}. For this, we introduce
a simple process calculus whose only non-trivial concurrency feature
is a principle of \emph{barrier synchronization}.
This is here understood intuitively as the single point of control where multiple
processes have to ``meet'' before continuing. This is one of the important
building blocks for concurrent and parallel systems~\cite{barrier:synchro:88}.

\comment{Review 3}{It is strange that nothing is said or explained in the Introduction (including the Related Work paragraph) about "barrier synchronization" - so prominent in the title - and why it would be interesting. Its single occurrence, is in the description of the outline of the paper. Overall, the class of processes interested seems to be of limited interest.
    - Explain the abbreviation B I T S (also) at its first occurrence in the Introduction.}%
{We added a paragraph about barrier synchronization and found a reference which toroughly investigate the principle. 
    But of course we also appeal here to the intuitive meaning of the notion in the concurrency folklore. For the ``limited interest'' we would more argue about ``limited expressivity''. The conclusion provides further discussion but the general intent is to study a specific notion in isolation, and then as a futur work to ``connect the dots''. We hope that the interest is still somewhat ensured but it is not for us to decide. We also added a sentence to explain the intent of the BITS rules.}

Combinatorics is about ``counting things'', and what we propose to count in our study is the
number of executions of processes wrt. their ``syntactic size''. This
is a symptom of the so-called ``combinatorial explosion'', a defining
characteristic of concurrency. As a first step, we show that counting executions of concurrent
processes is a difficult problem, even in the case of our calculus with limited expressivity. Thus, one important goal of our study is to
investigate interesting sub-classes for which the problem becomes
``less difficult''. To that end, we elaborate in this paper a
systematic decomposition of arbitrary processes, based on only four
rules: (B)ottom, (I)ntermediate, (T)op and (S)plit. Each rule explains
how to ``remove'' one node from the control graph of a process while
taking into account its contribution in the number of possible executions.
Indeed, one main feature of this BITS-decomposition is that it produces a symbolic integral formula to
solve the counting problem. Based on this procedure, we develop a
generic algorithm to generate process executions uniformly at
random. Since the algorithm is working on the control graph of
processes, it provides a way to statistically analyze processes
without constructing their state-space explicitly. In the worst case,
the algorithm cannot of course overcome the hardness of the problem 
it solves. However, depending on the
rules allowed during the decomposition, and also on the strategy
adopted, one can isolate interesting sub-classes wrt. the counting and
random sampling problem. We identify well-known
``structural'' sub-classes such as \emph{fork-join parallelism}~\cite{valiant:bsp94} and
\emph{asynchronous processes with promises}~\cite{LiskovS88}.
For some of these
sub-classes we develop dedicated and efficient counting and
random sampling algorithms. A large sub-class that we find
particularly interesting is what we call the ``BIT-decomposable''
processes, i.e. only allowing the three rules (B), (I) and (T) in the
decomposition.  The counting formula we obtain for such processes is
of a linear size (in the number of atomic actions in the processes, or
equivalently in the number of vertices in their control graph). We also
discuss informally the typical shape of ``BIT-free'' processes.

The outline of the paper is as follows. In Section~\ref{sec:lang} we introduce a minimalist calculus
of barrier synchronization. We show that the control graphs of processes expressed in this language are isomorphic
to arbitrary partially ordered sets (Posets) of atomic actions. From this we deduce our rather ``negative'' starting
point: counting executions in this simple language is intractable in the general case. In Section~\ref{sec:bit} we define
the BITS-decomposition, and we use it in Section~\ref{sec:randgen} to design a generic uniform random sampler.
In Section~\ref{sec:subclasses} we discuss various sub-classes of processes related to the proposed decomposition, and for some of them we explain how the
counting and random sampling problem can be solved efficiently. In Section~\ref{sec:bench} we propose an experimental study
of the algorithm toolbox discussed in the paper.

Note that some technical complement and proof details are deferred to an external ``companion'' document. Moreover we provide the full
source code developed in the realm of this work, as well as the benchmark scripts. All these complement informations are
available online\footnote{cf.~\url{https://gitlab.com/ParComb/combinatorics-barrier-synchro.git}}.

\subsection*{Related  work}

Our study intermixes viewpoints from concurrency theory, order-theory as well as combinatorics (especially enumerative combinatorics and random sampling). The \emph{heaps combinatorics} (studied in e.g.~\cite{trace:randgen:mfcs15}) provides a complementary interpretation of concurrent systems. One major difference is that this concerns ``true concurrent'' processes based on the trace monoid, while we rely on the alternative \emph{interleaving semantics}. A related uniform random sampler for networks of automata is presented in~\cite{automata:network:randge:concur17}. Synchronization is interpreted on words using a notion of ``shared letters''. This is very different from the ``structural'' interpretation as joins in the control graph of processes. For the generation procedure \cite{trace:randgen:mfcs15} requires the construction of a ``product automaton'', whose size grows exponentially in the number of ``parallel'' automata. By comparison, all the algorithms we develop are based on the control graph, i.e. the space requirement remains polynomial (unlike, of course, the time complexity in some cases). Thus, we can interpret this as a space-time trade-of between the two approaches. A related approach is that of investigating the combinatorics of \emph{lassos}, which is connected to the observation of state spaces through linear temporal properties. A uniform random sampler for lassos is proposed in~\cite{unif:mcmc:fase11}. The generation procedure takes place \emph{within} the constructed state-space, whereas the techniques we develop do not require this explicit construction. However lassos represent infinite runs whereas for now we only handle finite (or finite prefixes) of executions.

A coupling from the past (CFTP) procedure for the uniform random generation of linear extensions is described, with relatively sparse details, in~\cite{huber:dm06}. The approach we propose, based on the continuous embedding of Posets into the hypercube, is quite complementary. A similar idea is used in~\cite{BaMaWa18}
for the enumeration of Young tableaux using what is there called the \emph{density method}.  The paper~\cite{MCMC:modelcheck:tacas05} advocates the uniform random generation of executions
as an important building block for \emph{statistical model-checking}. A similar discussion is proposed
in~\cite{Sen:RandomTestingConcurrency:2007} for \emph{random testing}. The \emph{leitmotiv} in both cases is that
generating execution paths \emph{without} any bias is difficult. Hence a uniform random sampler is very likely to produce interesting and complementary tests, if comparing to other test generation strategies.

Our work can also be seen as a continuation of the \emph{algorithm and order} studies~\cite{algo:order:88} orchestrated by Ivan Rival in late 1980's only with powerful new tools available in the modern combinatorics toolbox. 

\section{Barrier synchronization processes}
\label{sec:lang}

The starting point of our study is the small process calculus described below.
\begin{definition}[Syntax of barrier synchronization processes]
We consider countably infinite sets $\mathcal{A}$ of (abstract) atomic actions, and $\mathcal{B}$ of barrier names. 
The set $\mathcal{P}$ of processes is defined by the following grammar:

  \begin{tabular}{lll}
$P,Q$ ::= & $0$ & (termination)\\
                                      & $\mid$ $\alpha.P$ & (atomic action and prefixing) \\
                                      & $\mid$  $\langle B \rangle P$ & (synchronization) \\
                                     & $\mid$  $\nu(B)P$ & (barrier and scope) \\
                                      & $\mid$  $P \parallel Q$ & (parallel)  \\
  \end{tabular} 
\end{definition}
\comment{Review 2/3}{(Review 2) The use of figure for defining formally concepts and functions is rather disappointing.

(Review 3)  I do not like having definitions in Tables or Figures. Incorporating them in the text (preferably with a definition number) makes it easier to find them back}{Most definitions have been removed from tables/figures and put into the text. One exception is the description of the BITS rules that are given in graphical terms (hence we think a Figure is adequate).}
The language has very few constructors and is purposely of limited expressivity. Processes in this language can only perform atomic actions, fork child processes and interact using a basic principle of \emph{synchronization barrier}. A very basic process is the following one: $$\nu(B)~ \left [ \mathsf{a_1}.\langle B \rangle ~\mathsf{a_2}.0 \parallel \langle B \rangle \mathsf{b_1}.0 \parallel \mathsf{c_1}.\langle B \rangle~0 \right ]$$

This process can initially perform the actions $\mathsf{a_1}$ and $\mathsf{c_1}$ in an arbitrary order. We then reach the state in which all the processes agrees to synchronize on $B$:
$$\nu(B)~\left [ \langle B \rangle ~\mathsf{a_2}.0 \parallel \langle B \rangle \mathsf{b_1}.0 \parallel \langle B \rangle~0 \right ]$$
The possible next transitions are:
$\xrightarrow{\mathsf{a_2}} \mathsf{b_1}.0 \xrightarrow{\mathsf{b_1}} 0,~\text{alternatively} \xrightarrow{\mathsf{b_1}} \mathsf{a_2}.0 \xrightarrow{\mathsf{a_2}} 0$

\noindent In the resulting states, the barrier $B$ has been ``consumed''.

The operational semantics below characterize processes
transitions of the form $P \xrightarrow{\alpha} P'$ in which $P$ can perform action
$\alpha$ to reach its (direct) derivative $P'$.
\begin{definition}[Operational semantics\label{def:sem}]
\vspace{-2ex}
\begin{center}
\begin{tabular}{c}
$\infer[\rulelbl{act}]{\alpha.P \xrightarrow{\alpha} P}{}$ \hspace{8pt} 
$\infer[\rulelbl{lpar}]{P \parallel Q \xrightarrow{\alpha} P' \parallel Q}{P \xrightarrow{\alpha} P'}$ \hspace{8pt}
$\infer[\rulelbl{rpar}]{P \parallel Q \xrightarrow{\alpha} P \parallel Q'}{Q \xrightarrow{\alpha} Q'}$
 \\[8pt]

$\infer[\rulelbl{lift}]{\nu(B) P \xrightarrow{\alpha} \nu(B) P'}{{\scriptstyle\fun{sync}_B(P)=Q} & {\scriptstyle\fun{wait}_B(Q)} & P \xrightarrow{\alpha} P'}$
\hspace{8pt}

$\infer[\rulelbl{sync}]{\nu(B) P \xrightarrow{\alpha} Q'}{{\scriptstyle\fun{sync}_B(P)=Q} & {\scriptstyle\lnot\fun{wait}_B(Q)} & Q \xrightarrow{\alpha} Q'}$ 

\end{tabular}
\end{center}
with:
$\begin{array}{ll}
\left [
\begin{array}{l}
{\scriptstyle\fun{sync}_B(0) = 0} \\
{\scriptstyle\fun{sync}_B(\alpha.P) = \alpha.P} \\
{\scriptstyle\fun{sync}_B(P \parallel Q) = \fun{sync}_B(P) \parallel \fun{sync}_B(Q)} \\
{\scriptstyle\fun{sync}_B(\nu(B)P) = \nu(B) P} \\
{\scriptstyle\forall C\neq B,~\fun{sync}_B(\nu(C)P) = \nu (C)~\fun{sync}_B(P)}\\
{\scriptstyle\fun{sync}_B(\langle B \rangle P) = P} \\
{\scriptstyle\forall C\neq B,~\fun{sync}_B(\langle C \rangle P) = \langle C \rangle P} 
\end{array}
\right .
& \left [
\begin{array}{l}
{\scriptstyle\fun{wait}_B(0) = \fun{false}} \\
{\scriptstyle\fun{wait}_B(\alpha.P) = \fun{wait}_B(P)} \\
{\scriptstyle\fun{wait}_B(P \parallel Q) = \fun{wait}_B(P) \lor \fun{wait}_B(Q)} \\
{\scriptstyle\fun{wait}_B(\nu(B)P) = \fun{false}} \\
{\scriptstyle\forall C\neq B,~\fun{wait}_B(\nu(C)P) = \fun{wait}_B(P)}\\
{\scriptstyle\fun{wait}_B(\langle B \rangle P) = \fun{true}} \\
{\scriptstyle\forall C\neq B,~\fun{wait}_B(\langle C \rangle P) = \fun{wait}_B(P)}\\
\end{array}
\right .
\end{array}$
\end{definition}
The rule $\rulelbl{sync}$ above explains the
synchronization semantics for a given barrier $B$. The rule is non-trivial given the
broadcast semantics of barrier synchronization. The definition is based on two auxiliary functions.
First, the function $\fun{sync}_B(P)$ produces a derivative process $Q$ in which all the possible
synchronizations on barrier $B$ in $P$ have been effected. If $Q$ has a sub-process that
cannot yet synchronize on $B$, then the predicate $\fun{wait}_B(Q)$ is true and the synchronization
on $B$ is said incomplete.  In this case the rule $\rulelbl{sync}$ does not apply,
however the transitions \emph{within} $P$ can still happen through $\rulelbl{lift}$.

\subsection{The control graph of a process}

We now define the notion of a (finite) execution of a process.

\begin{definition}[execution]
An \emph{execution} $\sigma$ of $P$ is a finite sequence
$\langle \alpha_1, \ldots, \alpha_n \rangle$
such that there exist a set of processes $P'_{\alpha_1},\ldots,P'_{\alpha_n}$
and a path
$P \xrightarrow{\alpha_1} P'_{\alpha_1} \ldots \xrightarrow{\alpha_n} P'_{\alpha_n}$
with $P'_{\alpha_n} \nrightarrow$ (no transition is possible from $P'_{\alpha_n}$).
\end{definition}

We assume that the occurrences of the atomic actions
in a process expression have all distinct labels, $\alpha_1, \ldots, \alpha_n$. This is allowed since the actions
 are uninterpreted in the semantics (cf.~Definition~\ref{def:sem}).
Thus, each action $\alpha$ in an execution $\sigma$ can be associated to a unique \emph{position}, which we denote by $\sigma(\alpha)$.
For example if $\sigma = \langle \alpha_1, \ldots,\alpha_k,\ldots, \alpha_n \rangle$, 
then $\sigma(\alpha_k) = k$.

The behavior of a process can be abstracted by considering
the \emph{causal ordering relation} wrt. its atomic actions.

\begin{definition}[cause, direct cause] \label{def:cause}
Let $P$ be a process. An action $\alpha$ of $P$ is said a \emph{cause}
of another action $\beta$, denoted by $\alpha < \beta$, iff
for any execution $\sigma$ of $P$ we have $\sigma(\alpha) < \sigma(\beta)$.
Moreover, $\alpha$ is a \emph{direct cause} of $\beta$, denoted by $\alpha \prec \beta$
iff $\alpha < \beta$ and there is no $\gamma$ such that $\alpha < \gamma < \beta$.
The relation $<$ obtained from $P$ is denoted by $\PO(P)$.
\end{definition}

\comment{Review 1 and 3}
{(Review 1) Definition 2 was a bit ambiguous, it needs to be made more formal.

(Review 3) - After Definition 1: for the sake of simplicity an injective labeling is assumed. It means that we focus on structure and events excluding multiple occurrences of the same action. Thus in my opinion this is a rather important (severe?) restriction, making it possible to have Theorem 2. (See also the first sentence of Section 7.)
- Related to Definition 2 and the text that follows. If in all executions of a process P, we have action a before b, but always with either action c or action d inbetween a and b, would you then call action a  a "direct cause" of b? There could be something wrong here, since the covering of a po is supposed to be its transitive reduction. Or you should explain that this situation does not occur for the processes under consideration (which is strange in view of Theorem 2).}
{There was indeed an issue with definition 2 (now definition 4) because the causal ordering was taken ``large'' and we forgot to specify that $\gamma$ had to be distinct from $\alpha$ and $\beta$. In fact, it is better to consider the ``strict'' ordering. We thus changed the definition and we think it is now correct. More importantly, we explain in a more precise way the ``injective labelling'' assumption before the definition. In fact, the semantics of the language we consider does not interpret the actions. This is an important aspect of our combinatorial interpretation, to consider the concurrency aspects in a purely structural way. This is also commented in the related work section because it is a departure from the ``automata-based'' approaches (that interpret actions with alphabets, and synchronization with actions sharing the same name).}

\begin{wrapfigure}[12]{r}{7cm}
\vspace*{-1cm}
\begin{center}
$\nu(B)\left[\begin{array}{l}
\alpha_1.\langle B\rangle \parallel \alpha_2.\langle B\rangle \parallel \ldots \parallel \alpha_n.\langle B \rangle \\
\langle B \rangle.\beta_1 \parallel \langle B \rangle.\beta_2 \parallel \ldots \parallel \langle B \rangle.\beta_n
\end{array} \right ]$ 

\vspace*{3mm}

\begin{tikzpicture}[node distance=40pt, baseline]
\node[draw,circle] (a1) {$\alpha_1$};
\node[draw,circle, right of=a1] (a2) {$\alpha_2$};
\node[right of=a2] (adot) {$\cdots$};
\node[draw,circle, right of=adot] (an) {$\alpha_n$};

\node[draw,circle,below of=a1] (b1) {$\beta_1$};
\node[draw,circle,below of=a2] (b2) {$\beta_2$};
\node[below of=adot] (bdot) {$\cdots$};
\node[draw,circle,below of=an] (bn) {$\beta_n$};

\draw[->] (a1) -- (b1);
\draw[->] (a1) -- (b2);
\draw[->] (a1) -- (bn);
\draw[->, dotted] (a1) -- (bdot);

\draw[->] (a2) -- (b1);
\draw[->] (a2) -- (b2);
\draw[->] (a2) -- (bn);
\draw[->, dotted] (a2) -- (bdot);

\draw[->,dotted] (adot) -- (b1);
\draw[->,dotted] (adot) -- (b2);
\draw[->,dotted] (adot) -- (bn);
\draw[->,dotted] (adot) -- (bdot);

\draw[->] (an) -- (b1);
\draw[->] (an) -- (b2);
\draw[->] (an) -- (bn);
\draw[->, dotted] (an) -- (bdot);
\end{tikzpicture}

\caption{\label{fig:ncycle} A process of size $2n$ and its control graph with $2n$ nodes and $n^2$ edges.}
\end{center}
\end{wrapfigure}
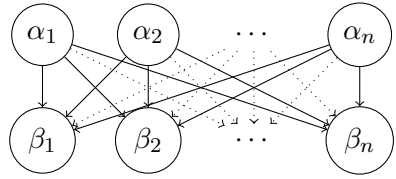

Obviously $\PO(P)$ is a \emph{partially ordered set} (poset) with \emph{covering} $\prec$,
capturing the \emph{causal ordering} of the actions of $P$.
The covering of a partial order is by construction an \emph{intransitive directed acyclic graph} (DAG),
hence the description of $\PO(P)$ itself is simply the transitive closure of the covering,
yielding $O(n^2)$ edges over $n$ elements. The worst case (maximizing the number of edges)
is a complete bipartite graph with two sets of $2n$ vertices connected by
$n^2$ edges (cf. Fig.~\ref{fig:ncycle}).\\
For most practical concerns we will only consider the covering, i.e.
the intransitive DAG obtained by the \emph{transitive reduction} of the order. It is possible to direclty
 construct this control graph, according to the following definition.

\comment{Review ? / Review 3}{
\textbf{Review ? (supplementary comment)}
I got stuck at Theorem 1:

Consider the process $P = \nu(B)\nu(C) [ <B><C>a.0 || <C><B>b.0 ]$ .
I assume it is a deadlock but I do not see why ctg(P) has a cycle,
which should be the case by Theorem 1.

I agree with others that the reader can get easily confused
since the notions around Theorem 1 are not clearly defined.

\textbf{Review 3} - The "definition" of control graph in Figure 3, given in the form of a computation, is hard to follow and seems incomplete. What is the meaning of the squiggly arrow?
}
{
  Indeed, the ctg definition is extracted from an algorithm implemented as a python program. We clearly went a little bit too far when trying to provide a concise definition. The thing is, it is not a difficult construction, only a little bit cumbersome to explain. In the previous definition, a subcase was missing for the singleton graphs (with vertices but no edges). 
  We rewrote the definition so that we explicitely build the vertices together with the edges. We believe the new definition is correct, and faithful to the original algorithm and Python implementation (except that in the latter, an exception is raised in case of a deadlock caused by a synchronization cycle).  The squiggly arrow is still present but hopefully in a way easier to understand.
 
Taking the new definition literaly, the construction is as follows on the example graph:

$\fun{ctg}(P) \begin{array}[t]{l} 
= \bigotimes_{\langle B \rangle} \bigotimes_{\langle  C \rangle} \left (  \fun{ctg}(\langle B \rangle \langle C \rangle a.0) \cup \fun{ctg}(\langle C \rangle \langle B \rangle b.0)\right) \\
= \bigotimes_{\langle B \rangle} \bigotimes_{\langle  C \rangle} \begin{array}[t]{l} 
\langle \{ \langle B \rangle, \langle C \rangle, a\}, \{(\langle B \rangle, \langle C \rangle), (\langle C \rangle, a)\}\rangle  \} \\
\cup \langle \{ \langle C \rangle, \langle B \rangle, b\}, \{(\langle C \rangle, \langle B \rangle), (\langle B \rangle, b)\}\rangle  
\end{array} \\
= \bigotimes_{\langle B \rangle} \bigotimes_{\langle  C \rangle} \langle \{\langle B\rangle, \langle C \rangle, a, b\}, \{(\langle B \rangle, \langle C \rangle), (\langle C \rangle, a),(\langle C \rangle, \langle B \rangle), (\langle B \rangle, b) \}\rangle \\
= \bigotimes_{\langle B \rangle} \langle \{\langle B\rangle, a, b\}, \{(\langle B \rangle, \langle B \rangle), (\langle B \rangle, a), (\langle B \rangle, b) \}\rangle \\
= \langle \{a, b\}, \{(\langle B \rangle, \langle B \rangle), (\langle B \rangle, a), (\langle B \rangle, b) \}\rangle
\end{array}$

Note that in the final step, the barrier $\langle B \rangle$ cannot be removed because of the self-loop. So there are two witnesses of the fact that the construction failed: there is still a barrier name in the process, and there is a cycle.

Now consider, the alternative:

$\fun{ctg}(P) \begin{array}[t]{l} 
= \bigotimes_{\langle B \rangle} \bigotimes_{\langle  C \rangle} \left (  \fun{ctg}(\langle B \rangle \langle C \rangle a.0) \cup \fun{ctg}(\langle B \rangle \langle C \rangle b.0)\right) \\
= \bigotimes_{\langle B \rangle} \bigotimes_{\langle  C \rangle} \begin{array}[t]{l} 
\langle \{ \langle B \rangle, \langle C \rangle, a\}, \{(\langle B \rangle, \langle C \rangle), (\langle C \rangle, a)\}\rangle  \} \\
\cup \langle \{ \langle B \rangle, \langle C \rangle, b\}, \{(\langle B \rangle, \langle C \rangle), (\langle C \rangle, b)\}\rangle  
\end{array} \\
= \bigotimes_{\langle B \rangle} \bigotimes_{\langle  C \rangle} \langle \{\langle B\rangle, \langle C \rangle, a, b\}, \{(\langle B \rangle, \langle C \rangle), (\langle C \rangle, a), (\langle C \rangle, b) \}\rangle \\
= \bigotimes_{\langle B \rangle} \langle \{\langle B\rangle, a, b\}, \{(\langle B \rangle, a), (\langle B \rangle, b) \}\rangle \\
= \langle \{a, b\}, \emptyset \rangle
\end{array}$

The graph with only two unrelated vertices and no edge is the correct construction.

We also added a pointer to the python notebook describing the construction with more details, including the python implementation and various examples. Of course we do not expect the reviewer to read the notebook, but we think it is a nice complement.
In the paper, we do not wish to add more details about this construction, which ultimately remains quite simple, of course when there is no error in the presentation.}

\begin{definition}[\label{fig:cover}Construction of control graphs]
Let $P$ be a process term. Its control graph is $\fun{ctg}(P)=\langle V, E\rangle$, constructed inductively as follows:

$\left [
\begin{array}{l|l}
\fun{ctg}(0) = \langle \emptyset, \emptyset \rangle 
& \phantom{i} \fun{ctg}(\nu(B)P) = \bigotimes_{\langle B\rangle} \fun{ctg}(P) \\[2pt]
\begin{array}{l}
\fun{ctg}(\alpha.P) = \alpha \leadsto \fun{ctg}(P)\\
\fun{ctg}(\langle B \rangle P) = \langle B \rangle \leadsto \fun{ctg}(P) \phantom{i}
\end{array}
& \phantom{i} \begin{array}[t]{l} \fun{ctg}(P \parallel Q) =  \fun{ctg}(P) \cup \fun{ctg}(Q) \\
  \text{  with } \langle V_1,E_1 \rangle \cup \langle V_2,E_2 \rangle = \langle V_1 \cup V_2, E_1 \cup E_2 \rangle 
\end{array}\\[2pt]
\end{array}
\right .$

with $\left\{\begin{array}{l}
x \leadsto \langle V, E \rangle = \langle V\cup \{x\}, \{(x,y) \mid  y \in \fun{srcs}(E) \lor (E=\emptyset \land y \in V)\} \rangle\\
\fun{srcs}(E) = \{ y \mid (y,z) \in E \land \nexists x,~(x,y)\in E \} \\
\bigotimes_{\langle B \rangle} \langle V, E\rangle = \langle V \setminus \{\langle B \rangle\}, E \setminus \begin{array}[t]{l}
 \{ (x,y) \mid x\neq y \land (x=\langle B \rangle \lor y = \langle B) \rangle \} \\
 \cup ~ \{ (\alpha,\beta) \mid \{(\alpha,\langle B \rangle),~(\langle B \rangle,\beta)\} \subseteq E \} \rangle
\end{array}
\end{array} \right .$
\end{definition}
Given a control graph $\Gamma$, the notation $x\leadsto \Gamma$ corresponds to prefixing the graph by a single atomic action. The set $\fun{srcs}(E)$ corresponds to the \emph{sources} of the edges in $E$, i.e. the vertices without an incoming edge.
 And $\bigotimes_{\langle B \rangle} \Gamma$ removes an explicit barrier node and connect all the processes ending in $B$ to the processes starting from it. In effect, this realizes the synchronization described by the barrier $B$.
We illustrate the construction on a simple process below:

\begin{small}
$\begin{array}[t]{l}  \fun{ctg}(\nu(B)\nu(C) [ \langle B\rangle \langle C\rangle a.0 || \langle B\rangle \langle C\rangle b.0 ]) \\
 = \bigotimes_{\langle B \rangle} \bigotimes_{\langle  C \rangle} \left (  \fun{ctg}(\langle B \rangle \langle C \rangle a.0) \cup \fun{ctg}(\langle B \rangle \langle C \rangle b.0)\right) \\
 = \bigotimes_{\langle B \rangle} \bigotimes_{\langle  C \rangle} \begin{array}[t]{l} 
 \langle \{ \langle B \rangle, \langle C \rangle, a\}, \{(\langle B \rangle, \langle C \rangle), (\langle C \rangle, a)\}\rangle  \} \\
 \cup \langle \{ \langle B \rangle, \langle C \rangle, b\}, \{(\langle B \rangle, \langle C \rangle), (\langle C \rangle, b)\}\rangle  
 \end{array} \\
 = \bigotimes_{\langle B \rangle} \bigotimes_{\langle  C \rangle} \langle \{\langle B\rangle, \langle C \rangle, a, b\}, \{(\langle B \rangle, \langle C \rangle), (\langle C \rangle, a), (\langle C \rangle, b) \}\rangle \\
 = \bigotimes_{\langle B \rangle} \langle \{\langle B\rangle, a, b\}, \{(\langle B \rangle, a), (\langle B \rangle, b) \}\rangle \\
 = \langle \{a, b\}, \emptyset \rangle
 \end{array}$
 \end{small}
 
The graph with only two unrelated vertices and no edge is the correct construction. Now, slightly changing the process we see how the construction
fails for deadlocked processes.

 \begin{small}
 $  \begin{array}[t]{l}
 \fun{ctg}(P)= \bigotimes_{\langle B \rangle} \bigotimes_{\langle  C \rangle} \left (  \fun{ctg}(\langle B \rangle \langle C \rangle a.0) \cup \fun{ctg}(\langle C \rangle \langle B \rangle b.0)\right) \\
 = \bigotimes_{\langle B \rangle} \bigotimes_{\langle  C \rangle} \begin{array}[t]{l} 
 \langle \{ \langle B \rangle, \langle C \rangle, a\}, \{(\langle B \rangle, \langle C \rangle), (\langle C \rangle, a)\}\rangle  \} \\
 \cup \langle \{ \langle C \rangle, \langle B \rangle, b\}, \{(\langle C \rangle, \langle B \rangle), (\langle B \rangle, b)\}\rangle  
 \end{array} \\
 = \bigotimes_{\langle B \rangle} \bigotimes_{\langle  C \rangle} \langle \{\langle B\rangle, \langle C \rangle, a, b\}, \{(\langle B \rangle, \langle C \rangle), (\langle C \rangle, a),(\langle C \rangle, \langle B \rangle), (\langle B \rangle, b) \}\rangle \\
 = \bigotimes_{\langle B \rangle} \langle \{\langle B\rangle, a, b\}, \{(\langle B \rangle, \langle B \rangle), (\langle B \rangle, a), (\langle B \rangle, b) \}\rangle \\
 = \langle \{a, b\}, \{(\langle B \rangle, \langle B \rangle), (\langle B \rangle, a), (\langle B \rangle, b) \}\rangle
\end{array}$
\end{small}
 
In the final step, the barrier $\langle B \rangle$ cannot be removed because of the self-loop. So there are two witnesses of the fact that the construction failed: there is still a barrier name in the process, and there is a cycle in the resulting graph.

\begin{theorem}\label{thm:process:dag} Let $P$ be a process, then $P$ has a deadlock iff $\fun{ctg}(P)$ has a cycle.
 Moreover, if $P$ is deadlock-free
 (hence it is a DAG) then $(\alpha,\beta) \in \fun{ctg}(P)$ iff $\alpha \prec \beta$ (hence the DAG is intransitive).
\end{theorem}
\begin{proof}[idea]
The proof is not difficult but slightly technical. The idea is to extend the notion of execution to go ``past'' deadlocks,
 thus detecting cycles in the causal relation. The details are given in companion document. \hfill \qed
\end{proof}

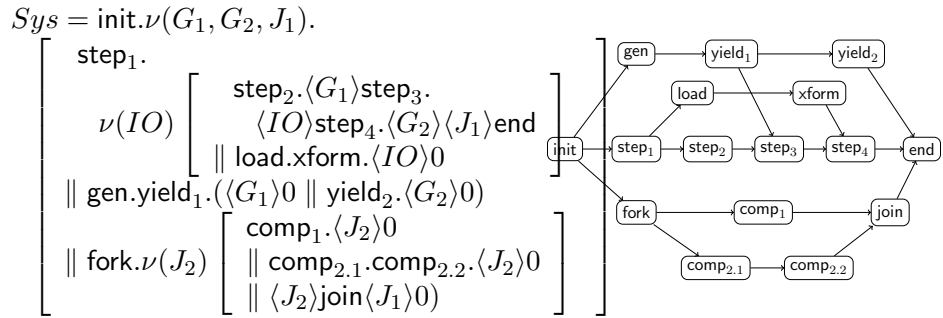
\begin{figure}[ht]
\begin{center}
\begin{tabular}{lr}
\begin{minipage}{0.55\textwidth}
$\begin{array}{l}
Sys = \mathsf{init}.\nu(G_1, G_2, J_1).\\
  \hspace{8pt}  \left [ \begin{array}{l}
    \begin{array}{l}
    \mathsf{step}_1. \\
     \hspace{8pt} \nu(IO)\left [\begin{array}{l}
       \hspace{8pt} \mathsf{step}_2.\langle{G_1}\rangle \mathsf{step}_3. \\
       \hspace{16pt}\langle{IO}\rangle \mathsf{step}_4.\langle{G_2}\rangle \langle J_1 \rangle \mathsf{end} \\
       \parallel \mathsf{load}.\mathsf{xform}.\langle{IO}\rangle 0
       \end{array} \right ] \\
     \end{array} \\
       \parallel  \mathsf{gen}.\mathsf{yield}_1.(\langle{G_1}\rangle 0 \parallel \mathsf{yield}_2.\langle{G_2}\rangle 0) \\
       \parallel \mathsf{fork}.\nu(J_2)\left[ \begin{array}{l}
                                              \mathsf{comp}_1.\langle J_2 \rangle 0 \\
                                              \parallel \mathsf{comp}_{2.1}. \mathsf{comp}_{2.2}.\langle J_2 \rangle 0\\
                                              \parallel \langle J_2\rangle \mathsf{join} \langle J_1 \rangle 0) \end{array}  \right]
       \end{array} \right]
\end{array}$
\end{minipage} &
\begin{minipage}{0.45\textwidth}
\scalebox{0.6}{\begin{tikzpicture}[act/.style={rectangle, draw, rounded corners}, node distance=45pt]
\node[act] (init) {$\mathsf{init}$};
\node[act, right of=init] (step1) {$\mathsf{step}_1$};
\node[act, above of=step1, node distance=60pt] (gstart) {$\mathsf{gen}$};
\draw[->] (init) -- (gstart);
\draw[->] (init)-- (step1);
\node[act, right of=step1] (step2) {$\mathsf{step}_2$};
\draw[->] (step1) -- (step2);
\node[act, right of=step2] (step3) {$\mathsf{step}_3$};
\draw[->] (step2) -- (step3);
\node[act, right of=step3] (step4) {$\mathsf{step}_4$};
\draw[->] (step3) -- (step4);
\node[act, right of=step4] (end) {$\mathsf{end}$};
\draw[->] (step4) -- (end);
\node[act, right of=gstart, node distance=60pt] (yield1) {$\mathsf{yield}_1$};
\draw[->] (gstart) -- (yield1);
\draw[->] (yield1) -- (step3);
\node[act, right of=yield1, node distance=80pt] (yield2) {$\mathsf{yield}_2$};
\draw[->] (yield1) -- (yield2);
\draw[->] (yield2) -- (end);
\node[act, above right of=step1, node distance=50pt] (load) {$\mathsf{load}$};
\draw[->] (step1) -- (load);
\node[act, right of=load, node distance=80pt] (xform) {$\mathsf{xform}$};
\draw[->] (load) -- (xform);
\draw[->] (xform) -- (step4);
\node[act, below of=step1, node distance=40pt] (fork) {$\mathsf{fork}$};
\node[right of=fork, node distance=50pt] (child) {};
\node[above of=child, node distance=0] (comp0) {};
\node[act, right of=comp0, node distance=30pt] (comp1) {$\mathsf{comp}_1$};
\node[act, below of=child, node distance=35pt] (comp2_1) {$\mathsf{comp}_{2.1}$};
\node[act, right of=comp2_1, node distance=65pt] (comp2_2) {$\mathsf{comp}_{2.2}$};
\node[act, right of=fork, node distance=160pt] (join) {$\mathsf{join}$};
\draw[->] (init) -- (fork);
\draw[->] (fork) -- (comp1);
\draw[->] (fork) -- (comp2_1);
\draw[->] (comp1) -- (join);
\draw[->] (comp2_1) -- (comp2_2);
\draw[->] (comp2_2) -- (join);
\draw[->] (join) -- (end);
\end{tikzpicture}}
\end{minipage}
\end{tabular}
\end{center}

\caption{\label{fig:example} An example process with barrier synchronizations (left) and its control graph (right). The process is of size 16 and
 it has exactly 1975974 possible executions.}

\end{figure}

In Fig.~\ref{fig:example} (left) we describe a system $Sys$ written in the proposed language,
together with the covering of $\PO(Sys)$, i.e. its control graph (right). We also indicate
the number of its possible executions, a question we address next.


\subsection{The counting problem}

One may think that in such a simple setting, any behavioral property, such as the counting problem that interests us, could be analyzed  
efficiently e.g. by a simple induction on the syntax. However, the devil is well hidden inside the box because of the following fact.

\begin{theorem}
Let $U$ be a partially ordered set. Then there exists a barrier synchronization process $P$ such that $\PO(P)$ is isomorphic to $U$.
\end{theorem}

\begin{proof} (sketch).
Consider $G$ the (intransitive) covering DAG of a poset $U$. 
We suppose each vertex of $G$ to be uniquely identified
by a label ranging over $\alpha_1, \alpha_2,\ldots, \alpha_n$.
The objective is to associate to each such vertex labeled $\alpha$ 
a process expression $P_\alpha$. The construction is done \emph{backwards},
starting from the \emph{sinks} (vertices without outgoing edges) of $G$
and bubbling-up until its \emph{sources} (vertices without incoming edges).  

There is a single rule to apply, considering a vertex labeled $\alpha$
whose children have already been processed, i.e. in a situation depicted as follows:

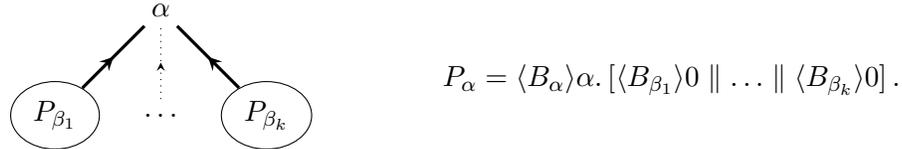
\begin{tabular}{ll}
\begin{tikzpicture}[baseline=-3.5ex, node distance=40pt]
\begin{scope}[decoration={
    markings,
    mark=at position 0.5 with {\arrow{stealth}}}
    ] 

\node (beta) {$\alpha$};
\node[ellipse,below of=beta] (alpha2) {$\ldots$};
\draw[postaction={decorate}, dotted] (alpha2) -- (beta);
\node[ellipse,draw,left of=alpha2] (alpha1) {$P_{\beta_1}$};
\draw[postaction={decorate}, very thick] (alpha1) -- (beta);
\node[ellipse,draw,right of=alpha2] (alphak) {$P_{\beta_k}$};
\draw[postaction={decorate}, very thick] (alphak) -- (beta);

\end{scope}
\end{tikzpicture}
 &

\begin{minipage}{0.7\textwidth}
\phantom{In this situation:}
\[ P_\alpha = 
\langle B_\alpha \rangle  \alpha. \left [ \langle B_{\beta_1} \rangle 0 \parallel \ldots \parallel \langle B_{\beta_k} \rangle 0 \right ]. 
\]
\end{minipage}
\end{tabular}

In the special case $\alpha$ is a sink we simply define $P_\alpha = \langle B_\alpha \rangle \alpha.0$.
In this construction it is quite obvious that $\alpha \prec \beta_i$ for each of the $\beta_i$'s, 
provided the barriers $B_\alpha,B_{\beta_1},\ldots,B_{\beta_k}$ are defined somewhere in the outer scope.

At the end we have a set of processes $P_{\alpha_1}, \ldots, P_{\alpha_n}$ associated to the vertices
of $G$ and we finally define
$ P = \nu(B_{\alpha_1}) \ldots \nu(B_{\alpha_n}) \left [ 
  P_{\alpha_1} \parallel \ldots \parallel P_{\alpha_n} \right ]$.

That $\PO(P)$ has the same covering as $U$ is a simple consequence of the construction. \hfill \qed
\end{proof}


\begin{corollary} \label{thm:linext:run}
Let $P$ be a non-deadlocked process. Then $\langle \alpha_1,\ldots,\alpha_n\rangle$  is an execution of $P$ 
if it is a linear extension of $\PO(P)$. 
Consequently, the number of executions of $P$ is equal to the number of linear extensions of $\PO(P)$.
\end{corollary}

We now reach our ``negative'' result that is the starting point of the rest of the paper: there is no
efficient algorithm to count the number of executions, even for such simplistic barrier processes.
\begin{corollary}
Counting the number of executions of a (non-deadlocked) barrier synchronization process is $\sharp{P}$-complete\footnote{A function $f$ is
 in $\sharp$P if there is a polynomial-time non-deterministic Turing machine $M$ such that for any instance $x$, $f(x)$ is the number
of executions of $M$ that accept $x$ as input. See~\url{https://en.wikipedia.org/wiki/\%E2\%99\%AFP-complete}}.
\end{corollary}
This is a direct consequence of~\cite{BW91} since counting executions of processes boils down to counting linear extensions
 in (arbitrary) posets.

\section{A generic decomposition scheme and its (symbolic) counting algorithm}
\label{sec:bit}

We describe in this section a generic (and symbolic) solution to the counting problem, based on a systematic
decomposition of finite Posets (thus, by Theorem~\ref{thm:process:dag}, of process expressions)
through their covering DAG (i.e. control graphs).

\subsection{Decomposition scheme}

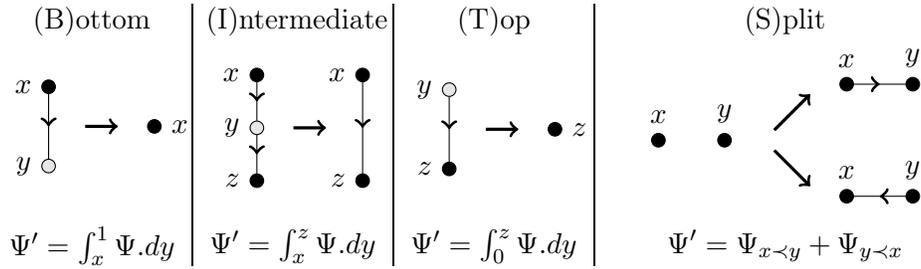
\begin{figure}[ht]
\begin{center}
\begin{tabular}{M | M | M | L}

(B)ottom& (I)ntermediate & (T)op & (S)plit  \\
\begin{tikzpicture}[node distance=30pt]
\node[shape=circle,draw,fill,label=left:$x$,scale=0.5] (x) {};
\node[shape=circle,draw,fill=gray!20, below
  of=x,label=left:$y$,scale=0.5] (y) {};

\path[style=ordedge] (x) -- (y);

\node (lsepp) at ($(x)!0.5!(y)$) {};
\node[node distance=10pt, right of=lsepp] (lsep) {};
\node[right of=lsep, node distance=20pt] (rsep) {};
\draw[->,very thick] (lsep) -- (rsep);

\node[shape=circle,draw,fill,label=right:$x$, right of=rsep, node
  distance=10pt,scale=0.5] (x') {};

\end{tikzpicture}
&
\begin{tikzpicture}[node distance=20pt]
\node[shape=circle,draw,fill,label=left:$x$,scale=0.5] (x) {};
\node[shape=circle,draw,fill=gray!20, below
  of=x,label=left:$y$,scale=0.5] (y) {}; \node[shape=circle,draw,fill,
  below of=y,label=left:$z$,scale=0.5] (z) {};

\path[style=ordedge] (x) -- (y); \path[style=ordedge] (y) -- (z);

\node[right of=y, node distance=10pt] (lsep) {};
\node[right of=lsep] (rsep) {};
\draw[->,very thick] (lsep) -- (rsep);

\node[right of=rsep, node distance=10pt] (y') {};
\node[shape=circle,draw,fill,label=left:$x$, above of=y',scale=0.5] (x') {};
\node[shape=circle,draw,fill, below of=y',label=left:$z$,scale=0.5] (z') {};

\draw[style=ordedge] (x') -- (z') ;
\end{tikzpicture}
&
\begin{tikzpicture}[node distance=30pt]
\node[shape=circle,draw,fill=gray!20,label=left:$y$,scale=0.5] (y) {};
\node[shape=circle,draw,fill, below of=y,label=left:$z$,scale=0.5] (z) {};

\path[style=ordedge] (y) -- (z);

\node (lsepp) at ($(z)!0.5!(y)$) {};
\node[node distance=10pt, right of=lsepp] (lsep) {};
\node[right of=lsep, node distance=20pt] (rsep) {};
\draw[->,very  thick] (lsep) -- (rsep);

\node[right of=rsep, node distance=10pt] (y') {};
\node[shape=circle,draw,fill, right of=rsep, node
  distance=10pt,label=right:$z$,scale=0.5] (z') {};
\end{tikzpicture}
&
\begin{tikzpicture}[node distance=20pt]
\node[shape=circle,draw,fill,label=above:$x$,scale=0.5] (x) {};
\node[shape=circle,draw,fill, right
  of=x,label=above:$y$,scale=0.5, node distance=25pt] (y) {}; 

\node[right of=y, node distance=15pt] (lsep) {};
\node[above right of=lsep, node distance=30pt] (rsep1) {};
\draw[->,very thick] (lsep) -- (rsep1);

\node[right of=rsep1, shape=circle,draw,fill,label=above:$x$,scale=0.5, node distance=10pt] (x1) {};
\node[shape=circle,draw,fill, right
  of=x1,label=above:$y$,scale=0.5, node distance=25pt] (y1) {}; 
\draw[style=ordedge] (x1) -- (y1) ;

\node[below right of=lsep, node distance=30pt] (rsep2) {};
\draw[->,very thick] (lsep) -- (rsep2);

\node[right of=rsep2, shape=circle,draw,fill,label=above:$x$,scale=0.5, node distance=10pt] (x2) {};
\node[shape=circle,draw,fill, right
  of=x2,label=above:$y$,scale=0.5, node distance=25pt] (y2) {}; 
\draw[style=ordedge] (y2) -- (x2) ;

\end{tikzpicture}
\\[4pt]
$\Psi' = \int^1_x \Psi . dy$
&
$\Psi' = \int^z_x \Psi . dy$
&
$\Psi' = \int^z_0 \Psi . dy$
&
$\displaystyle\Psi' = \Psi_{x \prec y} + \Psi_{y \prec x}$
\end{tabular}
\caption{\label{fig:decomposition} The BITS-decomposition and the construction of the counting formula.}
\end{center}
\end{figure}

In Fig.~\ref{fig:decomposition} we introduce the four decomposition rules that define
the BITS-decomposition. The first three rules are somehow straightforward.
The (B)-rule (resp. (T)-rule) allows to consume a node with no outgoing
(resp. incoming) edge and one incoming (resp. outgoing) edge. In a way, these two rules consume the
``pending'' parts of the DAG. The (I)-rule allows to
consume a node with exactly one incoming and outgoing edge. The final
(S)-rule takes two incomparable nodes $x,y$ and decomposes
 the DAG in two variants: the one for $x\prec y$ and the one for the converse $y\prec x$.

We now discuss the main interest of the decomposition: the incremental construction of an integral
formula that solves the counting problem. The calculation is governed by the equations specified below the rules in Fig.~\ref{fig:decomposition},
 in which the current formula $\Psi$ is updated  according to the definition of $\Psi'$ in the equations. 

\begin{theorem}\label{thm:bitc:integral}
The numerical evaluation of the integral formula built by the BITS-decomposition
yields the number of linear extensions of the corresponding Poset.
Moreover, the applications of the BITS-rules are confluent, in the sense that all the
sequences of (valid) rules reduce the DAG to an empty graph\footnote{At the end of the decomposition, the DAG is in fact reduced to a single node, which is removed by an integration between $0$ and $1$.}.
\end{theorem}

\comment{(Review 3)}
{- The BITS decomposition and the associated construction of the counting symbolic formulae are given in Figure 4 with as good as no further explanation of the application or justification for the formulae.}
{We decided to adopt a ``top-down'' presentation of the decomposition. Hence, in this section we only introduce the rules in an informal way, and we rely on the example to give a good intuition about it. Since the decomposition is only about removing nodes in a graph, it is not a complex matter. However, the justification of  the integral computations rely on classical combinatorics results, but based on a non-trivial mathematical theory for readers not familiar with the field. Thus, we postponed the justification to section 3.2. We reformulated a little bit section 3.1 to insist on this fact.}

The precise justification of the integral computation and  the proof for the theorem above are postponed to Section~\ref{sec:hypercube} below. We first consider an example.
\begin{example}\label{ex:bitc}Illustrating the BITS-decomposition scheme.
  \begin{center}
 \resizebox{.9\textwidth}{!}{
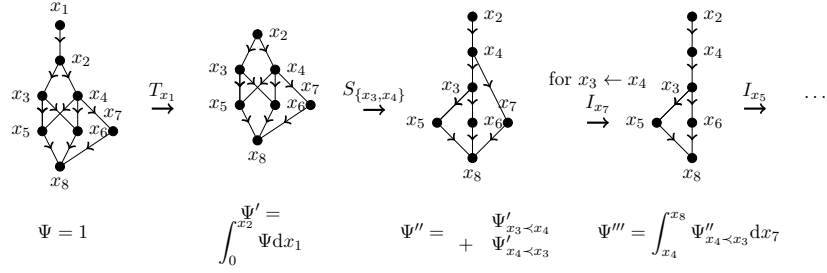
\begin{tabular}{MSMSMSMAA}
\begin{tikzpicture}[node distance=40pt]
  \node[shape=circle,draw,fill,label=above:$x_1$,scale=0.5] (1) {};
  \node[shape=circle,draw,fill,label=right:$x_2$,scale=0.5, below of=1] (3) {};
  \node[left of=3,node distance=10pt] (b1) {};
  \node[right of=3,node distance=10pt] (b2) {};
  \node[shape=circle,draw,fill,label=left:$x_3$,scale=0.5, below of=b1] (4) {};
  \node[shape=circle,draw,fill,label=right:$x_4$,scale=0.5, below of=b2] (5) {};
  \node[shape=circle,draw,fill,label=left:$x_5$,scale=0.5, below of=4] (6) {};
\node[shape=circle,draw,fill,label=right:$x_6$,scale=0.5, below of=5] (7) {};
\node[shape=circle,draw,fill,label=above:$x_7$,scale=0.5, right of=7] (8) {};
\node (mid67) at ($(6)!0.5!(7)$) {};
\node[shape=circle,draw,fill,label=below:$x_8$,scale=0.5, below of=mid67] (9) {};

\draw[style=ordedge] (1) -- (3);
\draw[style=ordedge] (3) -- (4);
\draw[style=ordedge] (3) -- (5);
\draw[style=ordedge] (4) -- (6);
\draw[style=my_ordedge] (4) -- (7);
\draw[style=my_ordedge] (5) -- (6);
\draw[style=ordedge] (5) -- (7);
\draw[style=ordedge] (6) -- (9);
\draw[style=ordedge] (7) -- (9);
\draw[style=ordedge] (4) -- (6);
\draw[style=ordedge] (5) -- (8);
\draw[style=ordedge] (8) -- (9);
\end{tikzpicture}
&
\begin{tikzpicture}
\node[node distance=20pt] (lsep) {};
\node[right of=lsep,node distance=20pt] (rsep) {};

\draw[->,very thick] (lsep) -- (rsep)  node [auto, midway] {$T_{x_1}$};
\end{tikzpicture}
&
\begin{tikzpicture}[node distance=40pt]
  \node[shape=circle,draw,fill,label=right:$x_2$,scale=0.5] (3) {};
  \node[left of=3,node distance=10pt] (b1) {};
  \node[right of=3,node distance=10pt] (b2) {};
  \node[shape=circle,draw,fill,label=left:$x_3$,scale=0.5, below of=b1] (4) {};
  \node[shape=circle,draw,fill,label=right:$x_4$,scale=0.5, below of=b2] (5) {};
  \node[shape=circle,draw,fill,label=left:$x_5$,scale=0.5, below of=4] (6) {};
\node[shape=circle,draw,fill,label=right:$x_6$,scale=0.5, below of=5] (7) {};
\node[shape=circle,draw,fill,label=above:$x_7$,scale=0.5, right of=7] (8) {};
\node (mid67) at ($(6)!0.5!(7)$) {};
\node[shape=circle,draw,fill,label=below:$x_8$,scale=0.5, below of=mid67] (9) {};

\draw[style=ordedge] (1) -- (3);
\draw[style=ordedge] (3) -- (4);
\draw[style=ordedge] (3) -- (5);
\draw[style=ordedge] (4) -- (6);
\draw[style=my_ordedge] (4) -- (7);
\draw[style=my_ordedge] (5) -- (6);
\draw[style=ordedge] (5) -- (7);
\draw[style=ordedge] (6) -- (9);
\draw[style=ordedge] (7) -- (9);
\draw[style=ordedge] (4) -- (6);
\draw[style=ordedge] (5) -- (8);
\draw[style=ordedge] (8) -- (9);
\end{tikzpicture}
&
\begin{tikzpicture}
\node[node distance=20pt] (lsep) {};
\node[right of=lsep,node distance=20pt] (rsep) {};

\draw[->,very thick] (lsep) -- (rsep) node [auto, midway] {$S_{\{x_3, x_4\}}$};
\end{tikzpicture}
&
\begin{tikzpicture}[node distance=40pt]
  \node[shape=circle,draw,fill,label=right:$x_2$,scale=0.5] (3) {};
  \node[shape=circle,draw,fill,label=right:$x_4$,scale=0.5, below of=3] (5) {};
  \node[shape=circle,draw,fill,label=left:$x_3$,scale=0.5, below of=5] (4) {};
  \node[shape=circle,draw,fill,label=right:$x_6$,scale=0.5, below of=4] (7) {};
  \node[shape=circle,draw,fill,label=left:$x_5$,scale=0.5, left of=7] (6) {};
\node[shape=circle,draw,fill,label=above:$x_7$,scale=0.5, right of=7] (8) {};
\node[shape=circle,draw,fill,label=below:$x_8$,scale=0.5, below of=7] (9) {};

\draw[style=ordedge] (1) -- (3);
\draw[style=ordedge] (5) -- (4);
\draw[style=ordedge] (3) -- (5);
\draw[style=ordedge] (4) -- (6);
\draw[style=ordedge] (4) -- (7);
\draw[style=ordedge] (6) -- (9);
\draw[style=ordedge] (7) -- (9);
\draw[style=ordedge] (4) -- (6);
\draw[style=ordedge] (5) -- (8);
\draw[style=ordedge] (8) -- (9);
\end{tikzpicture}
  &
\begin{tikzpicture}
  \node[node distance=20pt] (lsep) {};
  \node[right of=lsep,node distance=20pt] (rsep) {};
  
  \draw[->,very thick] (lsep) -- (rsep) node [auto, midway] {$\begin{array}{c} \text{for } x_3\leftarrow x_4 \\ I_{x_7}\end{array}$};
\end{tikzpicture}
  &
\begin{tikzpicture}[node distance=40pt]
  \node[shape=circle,draw,fill,label=right:$x_2$,scale=0.5] (3) {};
  \node[shape=circle,draw,fill,label=right:$x_4$,scale=0.5, below of=3] (5) {};
  \node[shape=circle,draw,fill,label=left:$x_3$,scale=0.5, below of=5] (4) {};
  \node[shape=circle,draw,fill,label=right:$x_6$,scale=0.5, below of=4] (7) {};
  \node[shape=circle,draw,fill,label=left:$x_5$,scale=0.5, left of=7] (6) {};
  \node[shape=circle,draw,fill,label=below:$x_8$,scale=0.5, below of=7] (9) {};
  \node[left of=6, node distance=30pt] (b1) {};
  
\draw[style=ordedge] (1) -- (3);
\draw[style=ordedge] (5) -- (4);
\draw[style=ordedge] (3) -- (5);
\draw[style=ordedge] (4) -- (6);
\draw[style=ordedge] (4) -- (7);
\draw[style=ordedge] (6) -- (9);
\draw[style=ordedge] (7) -- (9);
\draw[style=ordedge] (4) -- (6);
\end{tikzpicture}
    &
\begin{tikzpicture}
\node[node distance=20pt] (lsep) {};
\node[right of=lsep,node distance=20pt] (rsep) {};

\draw[->,very thick] (lsep) -- (rsep) node [auto, midway] {$I_{x_5}$};
\end{tikzpicture}
&
\begin{tikzpicture}[node distance=40pt]
\node (a5) {$\dots$};
\end{tikzpicture}\\
$\Psi = 1$
  &
  &
    $\displaystyle \Psi' = \int_0^{x_2} \Psi \dd x_1$
  &
  &
\multicolumn{2}{l}{
    $\displaystyle\Psi'' = \begin{array}{ll}  & \Psi'_{x_3\prec x_4}\\ + &
                            \Psi'_{x_4 \prec x_3}\\
                          \end{array}$}
  &
  
    \multicolumn{2}{l}{
    $\displaystyle
    \Psi''' = \int_{x_4}^{x_8} \Psi''_{x_4 \prec x_3} \mathrm{d}x_7$}
  &
\end{tabular}
}
\end{center}
\end{example}

The DAG to decompose (on the left) is of size $8$ with nodes
$x_1,\ldots,x_8$. The decomposition is non-deterministic, multiple rules apply, e.g. we
could ``consume'' the node $x_7$ with the (I) rule. Also, the (S)plit rule is always enabled.
 In the example, we decide to first remove the node $x_1$ by an application of
the (T) rule. We then show an application of the (S)plit rule for the incomparable nodes $x_3$ and $x_4$.  
The decomposition should then be performed on two distinct DAGs: one for $x_3 \prec x_4$ and the
other one for $x_4 \prec x_3$. We illustrate the second choice, and we further eliminate the nodes $x_7$ then $x_5$ using
 the (I) rule, etc. Ultimately all the DAGs are decomposed and we obtain the following integral computation:
\begin{scriptsize}
\begin{align*}
	\Psi & = \int_{x_2=0}^1 \int_{x_4=x_2}^1 \int_{x_3=x_4}^1 \int_{x_6=x_3}^1 \int_{x_8=x_6}^1 \int_{x_5= x_3}^{x_8} \int_{x_7=x_4}^{x_8} \\
		& \qquad   \left( \mathbbm{1}_{\mid x_4 \prec x_3} \cdot \int_{x_1=0}^{x_2} 1\cdot \dd x_1 
			+ \mathbbm{1}_{\mid x_3 \prec x_4}\cdot \int_{x_1=0}^{x_2}  1\cdot  \dd x_1 \right)  \dd x_7 \dd x_5 \dd x_8 \dd x_6 \dd x_3 \dd x_4 \dd x_2 
	 = \frac{8 + 6}{8!}. 
\end{align*}
\end{scriptsize}
The result means that there are exactly 14 distinct linear extensions in the example Poset.

\subsection{Embedding in the hypercube: the order polytope}
\label{sec:hypercube}

The justification of our decomposition scheme is based on the continuous embedding of posets
into the hypercube, as investigated in~\cite{stanley86}.

\begin{definition}[order polytope]
\begin{sloppypar}
  Let $P = (E, \prec)$ be a poset of size $n$.
  Let $C$ be the unit hypercube defined by $C = \left\{(x_1, \dots,
  x_n) \in \mathbb{R}^n ~ | ~ \forall i, ~0 \leq x_i \leq 1 \right\}$.
  For each constraint $ x_i \prec x_j \in P$ we define the convex subset
  $S_{i,j} = \left\{(x_1, \dots, x_n) \in \mathbb{R}^n ~ | ~ x_i \leq x_j
  \right\}$, i.e. one of the half spaces obtained by cutting
  $\mathbb{R}^n$ with the hyperplane $\left\{(x_1, \dots, x_n) \in
  \mathbb{R}^n ~ | ~ x_i - x_j = 0 \right\}$.
  Thus, the order polytope $C_P$ of $P$ is:
  $$ C_p = \bigcap_{x_i \prec x_j \in P} S_{i,j} \cap C$$
\end{sloppypar}
\end{definition}


Each linear extension, seen as total orders, can similarly be
embedded in the unit hypercube. Then, the order polytopes of the linear
extensions of a poset $P$ form a partition of the
Poset embedding $C_p$ as illustrated in Figure~\ref{fig:hypercube}.

\begin{figure}[ht]
\begin{center}
\newcommand{\Length}{-1}
\begin{tabular}{c | c | c}
\begin{tikzpicture}

\coordinate (O) at (0,0,0);
\coordinate (A) at (0.41*\Length,0.71*\Length,0.58*\Length);
\coordinate (B) at (-0.41*\Length,0.71*\Length,1.15*\Length);
\coordinate (C) at (-0.82*\Length,0,0.58*\Length);
\coordinate (D) at (0.41*\Length,-0.71*\Length,0.58*\Length);
\coordinate (E) at (0.82*\Length,0,1.15*\Length);
\coordinate (F) at (0,0,1.73*\Length);
\coordinate (G) at (-0.41*\Length,-0.71*\Length,1.15*\Length);

\draw[black,dashed] (O) -- (D);
\draw[black,dashed] (A) -- (O) -- (C);
\draw[black] (G) -- (D) -- (E) -- (A);
\draw[black] (F) -- (E);
\draw[black] (B) -- (F);
\draw[black] (B) -- (C);
\draw[black] (A) -- (B);
\draw[black] (C) -- (G);
\draw[black] (F) -- (G);

\node[below left of=A, node distance=8pt] (a) {$C_{(0,1,0)}$};
\node[right of=B, node distance=18pt] (a) {$B_{(1,1,0)}$};
\node[right of=C, node distance=18pt] (a) {$A_{(1,0,0)}$};
\node[left of=O, node distance=18pt] (a) {$O_{(0,0,0)}$};
\node[left of=D, node distance=18pt] (a) {$E_{(0,0,1)}$};
\node[left of=E, node distance=18pt] (a) {$D_{(0,1,1)}$};
\node[right of=F, node distance=18pt] (a) {$G_{(1,1,1)}$};
\node[right of=G, node distance=18pt] (a) {$F_{(1,0,1)}$};

\node[below of=C, node distance=24pt] (b) {};
\end{tikzpicture} \hspace*{0.5cm}
&
\hspace*{0.5cm} \begin{tikzpicture}

\coordinate (O) at (0,0,0);
\coordinate (A) at (0.41*\Length,0.71*\Length,0.58*\Length);
\coordinate (B) at (-0.41*\Length,0.71*\Length,1.15*\Length);
\coordinate (C) at (-0.82*\Length,0,0.58*\Length);
\coordinate (D) at (0.41*\Length,-0.71*\Length,0.58*\Length);
\coordinate (E) at (0.82*\Length,0,1.15*\Length);
\coordinate (F) at (0,0,1.73*\Length);
\coordinate (G) at (-0.41*\Length,-0.71*\Length,1.15*\Length);

\draw[black,dashed] (O) -- (D);
\draw[black,dashed] (A) -- (O) -- (C);
\draw[black] (G) -- (D) -- (E) -- (A);
\draw[black] (F) -- (E);
\draw[black] (B) -- (F);
\draw[black] (B) -- (C);
\draw[black] (A) -- (B);
\draw[black] (C) -- (G);
\draw[black] (F) -- (G);

\draw[blue, dashed] (O) -- (G);
\draw[blue, dashed] (O) -- (D);
\draw[blue, dashed] (O) -- (F);
\draw[blue] (D) -- (F) -- (G) -- cycle;

\node[left of=A, node distance=8pt] (a) {$C$};
\node[right of=B, node distance=6pt] (a) {$B$};
\node[right of=C, node distance=6pt] (a) {$A$};
\node[left of=O, node distance=8pt] (a) {$O$};
\node[left of=D, node distance=6pt] (a) {$E$};
\node[left of=E, node distance=6pt] (a) {$D$};
\node[right of=F, node distance=6pt] (a) {$G$};
\node[right of=G, node distance=6pt] (a) {$F$};

\node[below of=C, node distance=24pt] (b) {};
\end{tikzpicture}\hspace*{0.5cm}
&
\hspace*{0.5cm}\begin{tikzpicture}

\coordinate (O) at (0,0,0);
\coordinate (A) at (0.41*\Length,0.71*\Length,0.58*\Length);
\coordinate (B) at (-0.41*\Length,0.71*\Length,1.15*\Length);
\coordinate (C) at (-0.82*\Length,0,0.58*\Length);
\coordinate (D) at (0.41*\Length,-0.71*\Length,0.58*\Length);
\coordinate (E) at (0.82*\Length,0,1.15*\Length);
\coordinate (F) at (0,0,1.73*\Length);
\coordinate (G) at (-0.41*\Length,-0.71*\Length,1.15*\Length);

\draw[black,dashed] (O) -- (D);
\draw[black,dashed] (O) -- (C);
\draw[black] (G) -- (D) -- (E) -- (A);
\draw[black] (F) -- (E);
\draw[black] (B) -- (F);
\draw[black] (B) -- (C);
\draw[black] (A) -- (B);
\draw[black] (C) -- (G);
\draw[black] (F) -- (G);

\draw[blue, dashed] (O) -- (G);
\draw[blue, dashed] (O) -- (D);
\draw[blue, dashed] (O) -- (F);
\draw[blue] (D) -- (F) -- (G) -- cycle;

\draw[red, dashed] (D) -- (F) ;
\draw[red, dashed] (F) -- (E);
\draw[red] (F) -- (E) -- (D) ;
\draw[red, dashed, dash phase=15pt] (O) -- (F) ;
\draw[red, dashed, dash phase=15pt] (O) -- (D) ;
\draw[red, dashed] (O) -- (E) ;

\draw[green!80!black, dashed, dash phase=15pt] (O) -- (E);
\draw[green!80!black, dashed] (O) -- (A);
\draw[green!80!black, dashed] (O) -- (F);
\draw[green!80!black, dashed, dash phase=15pt] (E) -- (F);
\draw[green!80!black] (E) -- (A) -- (F);

\node[left of=A, node distance=8pt] (a) {$C$};
\node[right of=B, node distance=6pt] (a) {$B$};
\node[right of=C, node distance=6pt] (a) {$A$};
\node[left of=O, node distance=8pt] (a) {$O$};
\node[left of=D, node distance=6pt] (a) {$E$};
\node[left of=E, node distance=6pt] (a) {$D$};
\node[right of=F, node distance=6pt] (a) {$G$};
\node[right of=G, node distance=6pt] (a) {$F$};

\node[below of=C, node distance=24pt] (b) {};
\end{tikzpicture}\\

\end{tabular}

\end{center}
\caption{\label{fig:hypercube}
From left to right: the unit hypercube, the embedding of the total order $1 \prec 2 \prec 3$ and the embedding of the poset $P = (\{1,2,3\}, \{1 \prec 2 \})$ divided in its three linear extensions.
}
\end{figure}
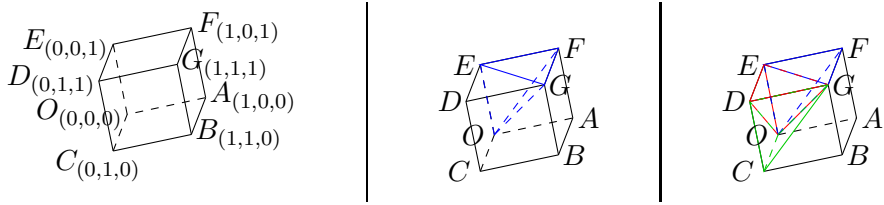

The number of linear extensions of a poset $P$, written $|\LE(P)|$, is then characterized as a volume in the embedding.

\begin{theorem}(\cite[Corollary 4.2]{stanley86})
\label{thm:stanley}
  Let $P$ be a Poset of size $n$ then its number of linear extensions $ |\LE(P)| = n! \cdot Vol(C_P)$
  where $Vol(C_P)$ is the volume, defined by the Lebesgue measure, of
  the order polytope $C_P$.
\end{theorem}

The integral formula introduced in the BITS-decomposition corresponds to the computation of $Vol(C_p)$,
 hence we may now give the key-ideas of Theorem~\ref{thm:bitc:integral}.
\begin{proof}[Theorem~\ref{thm:bitc:integral}, sketch]
  We begin with the (S)-rule. Applied on two incomparable elements $x$ and $y$, the rule partitions 
  the polytope in two regions: one for $x \prec y$ and the other for
   $y \prec x$. Obviously, the respective volume of the two disjoint regions must be added.\\
  We focus now on the (I)-rule.
  In the context of Lebesgue integration, the classic Fubini's theorem allows to compute
  the volume $V$ of a polytope $P$ as an iteration on integrals along each dimension,
  and this in all possible orders, which gives the confluence property. Thus,
\begin{footnotesize}
  \[  V = \int_{[0,1]^n} \mathbbm{1}_P(\mathbf{x}) \dd \mathbf{x}= \int_{[0,1]} \dots \int_{[0,1]} \mathbbm{1}_P((x, y, z, \dots)) \dd x \dd y \dd z \dots, \]
 \end{footnotesize}
 $\mathbbm{1}_P$ being the indicator function of $P$ such that
 $
   \displaystyle{ \mathbbm{1}_P((x, y, z, \dots)) = \prod_{\alpha \text{ actions}} \mathbbm{1}_{P_\alpha}(\alpha),}
  $
  with $P_\alpha$ the projection of $P$ on the dimension associated to $\alpha$. By convexity of
  $P$, the function $\mathbbm{1}_{P_y}$ is the indicator function of a segment $[x,z]$.
  So the following identity holds:
  $\int_P \mathbbm{1}_{P_y}(y) \dd y = \int_x^z \dd y$.
  Finally, the two other rules (T) and (B) are just special cases (taking $x=0$, alternatively $z=1$). \hfill\qed
\end{proof}

\begin{corollary}(Stanley~\cite{stanley86})
  The order polytope of a linear extension is a simplex
  and the simplices of the linear extensions are isometric,
  thus of the same volume.
\end{corollary}

\section{Uniform random generation of process executions}
\label{sec:randgen}

In this section we describe a generic algorithm for the uniform random
generation of executions of barrier synchronization processes. 
The algorithm is based on the BITS-decomposition and its embedding in the unit hypercube. It has two essential properties. First, it is directly working on the control graphs (equivalently on the corresponding poset), and thus does not require the explicit construction of the state-space of processes. Second, it generates possible executions of processes at random according to the uniform distribution. This is a guarantee that the sampling is not biased and reflects the actual behavior of the processes.

\algrenewcommand\algorithmicindent{1em}%
\begin{algorithm}[h]
  \caption{Uniform sampling of a simplex of the order polytope}
  \label{algo:randgen}
  \begin{algorithmic}[0]
    \Function{SamplePoint\footnotemark}{$\mathcal{I}=\int_a^b f(y_i) \,\dd y_i$}
    \State $C \gets \fun{eval}(\mathcal I)$ ~ ; ~ $U \gets \Unif(a,b)$
    \State $Y_i \gets$ the solution $t$ of $\int_a^t \frac{1}{C} f(y_i)\, \dd y_i= U$
    \If{$f$ is not a symbolic constant}
    \State \textsc{SamplePoint}$\left(f\{y_i \leftarrow
        Y_i\}\right)$
    \Else ~ \textbf{return} the $Y_i$'s
    \EndIf
    \EndFunction
  \end{algorithmic}
\end{algorithm}

\footnotetext{The Python/Sage implementation of the random sampler is available at the following location: \url{https://gitlab.com/ParComb/combinatorics-barrier-synchro/blob/master/code/RandLinExtSage.py}}

\comment{Review 2}
{Algorithm 1 is vague and need really much more explanations.}
{We totally agree with this comment. Thus, we completely reformulated the explanations of the algorithm, giving a lot more
details. In consequence, we had to remove some less commented parts of the paper so that we would fit the allowed 20 pages.

(Note that there is no further comment box beyond this point).}

\begin{sloppypar}
The starting point of Algorithm~\ref{algo:randgen} (cf. previous page) is a Poset over a set of points $\{x_1,\ldots,x_n\}$ (or equivalently its covering DAG). 
The decomposition scheme of Section~\ref{sec:bit} produces an integral formula $\mathcal{I}$ of the form $\int^1_0 F(y_n,\ldots,y_1) ~ \dd y_n\cdots \dd y_1$.
with $F$ a symbolic integral formula over the points $x_1,\ldots,x_n$. The $y$ variables represent a permutation of the poset points giving the order followed along the decomposition. Thus, the variable $y_i$ corresponds to the $i$-th removed point during the decomposition. 
We remind the reader that the evaluation of the formula $\mathcal{I}$ gives the number of linear extensions of the partial order.
Now, starting with the complete formula, the variables $y_1,~y_2,\ldots$ will be eliminated, in turn, in an ``outside-in'' way.
Algorithm~\ref{algo:randgen} takes place at the $i$-th step of the process. At this step, the considered formula is of the following form:
$$\int_a^b \underbrace{ \left (\int \cdots \int 1~\dd y_n \cdots \dd y_{i+1} \right )}_{f(y_i)}  \dd y_i.$$
Note that in the subformula $f(y_i)$ the variable $y_i$ may only occur (possibly multiple times) as an integral bound.
\end{sloppypar}

In the algorithm, the variable $C$ gets the result of the numerical computation of the integral $\mathcal{I}$ at the given step.
Next we draw (with $\Unif$) a real number $U$ uniformly at random between the integration bounds $a$ and $b$. Based on these two
intermediate values, we perform a numerical solving of variable $t$ in the integral formula corresponding to the \emph{slice} of the 
polytope along the hyperplan $y_i=U$. The result, a real number between $a$ and $b$, is stored in variable $Y_i$. The justification of
this step is further discussed in the proof sketch of Theorem~\ref{thm:randgen} below.

If there remains integrals in $\mathcal{I}$, the algorithm is applied recursively by substituting the variable $y_i$ in the integral bounds of $\mathcal{I}$ by the numerical value $Y_i$. If no integral remains, all the computed values $Y_i$'s are returned. As illustrated in Example~\ref{ex:randgen} below, this allows to select a specific linear extension in the initial partial ordering.  The justification of the algorithm is given by the following theorem.
\begin{theorem}\label{thm:randgen}
  Algorithm~\ref{algo:randgen} uniformly samples a point of the
  order polytope with a $\mathcal{O}(n)$ complexity in the number of
  integrations.
\end{theorem}

\begin{proof}[sketch] 
The problem is reduced to the uniform random sampling of a point $p$
in the order polytope. This is a classical problem about marginal densities that can be solved by slicing the polytope and evaluating
incrementally the $n$ continuous random variables associated to the coordinates of $p$.
More precisely, during the calculation of the volume of the polytope $P$, 
the last integration (of a monovariate polynomial $p(y)$) done from 0 to 1 corresponds to integrate the slices of $P$ according the last variable $y$.  
So, the polynomial $p(y)/\int_0^1 p(y)dy$ is nothing but the density function of the random variable $Y$ corresponding to the value of $y$.
Thus, we can generate $Y$ according to this density and fix it. When this is done,
we can inductively continue with the previous integrations to draw all 
the random variables associated to the coordinates of $p$.
The linear complexity of Algorithm~\ref{algo:randgen} follows from the fact that each partial integration deletes exactly one variable (which corresponds to one node). Of course at each step a possibly costly computation of the counting formula is required. \hfill \qed
\end{proof}

We now illustrate the sampling process based on Example~\ref{ex:bitc} (page~\pageref{ex:bitc}).

\begin{example}\label{ex:randgen}
  First we assume that the whole integral formula has already been
  computed. To simplify the presentation we only consider 
  (S)plit-free DAGs \emph{i.e.} decomposable without the
  (S) rule. Note that it would be easy to deal with the
  (S)plit rule: it is sufficient to uniformly choose one of the DAG
  processed by the (S)-rule w.r.t. their number of linear extensions.
  
  Thus we will run the example on the DAG of Example~\ref{ex:bitc}
  where the DAG corresponding to ``$x_4 \prec x_3$'' as been randomly chosen (with probability
  $\frac{8}{14}$) \emph{i.e.} the following formula holds:
  \begin{equation*}\label{eq:randgen}
    \int_0^1 \left (\int_{x_2}^1 \int_{x_4}^1 \int_{x_3}^1
    \int_{x_6}^1 \int_{x_4}^{x_8} \int_{x_3}^{x_8} \int_0^{x_2} \dd
    x_1 \dd x_5 \dd x_7 \dd x_8 \dd x_6 \dd x_3 \dd x_4 \right ) \dd x_2  = \frac{8}{8!}.
  \end{equation*}
In the equation above, the sub-formula between parentheses would be denoted by $f(x_2)$ in the
explanation of the algorithm.
  Now, let us apply the Algorithm~\ref{algo:randgen} to that formula
  in order to sample a point of the order polytope. In the first step the normalizing constant $C$ is equal to
  $\frac{8!}{8}$, we draw $U$ uniformly in $[0,1]$ and so we compute a
  solution of $\frac{8!}{8}\int_0^t \dots \dd x_2 = U$. That solution
  corresponds to the second coordinate of a the point we are
  sampling. And so on, we obtain values for each of the coordinates:
  \[
    \left\{
      \begin{array}{lllllll}
        X_1 = 0.064\dots, &\quad& X_2 = 0.081\dots,
        &\quad& X_3 = 0.541\dots, &\quad& X_4 = 0.323\dots,\\
        X_5 = 0.770\dots, &\quad& X_6 = 0.625\dots,
        &\quad& X_7 = 0.582\dots, &\quad& X_8 = 0.892\dots\\
      \end{array}
    \right.
  \]
  These points belong to a simplex of the order polytope. To find the
  corresponding linear extension we compute the rank of that vector
  \emph{i.e.} the order induced by the values of the coordinates
  correspond to a linear extension of the original DAG:
  \[(x_1,x_2,x_4,x_3,x_7,x_6,x_5,x_8).\]
This is ultimately the linear extension returned by the algorithm.
\end{example}



\section{Classes of processes that are BIT-decomposable\\(or not)}
\label{sec:subclasses}

Thanks to the BITS decomposition scheme, we can generate a counting formula
for any (deadlock-free) process expressed in the
barrier synchronization calculus, and derive from it a dedicated uniform random sampler.
However the (S)plit rule generates two summands, thus if we cannot find common calculations
between the summands the resulting formula can grow exponentially in the size of the concerned process.
If we avoid splits in the decomposition, then the counting formula remains of linear size.
This is, we think, a good indicator that the subclass of so-called ``BIT-decomposable'' processes
is worth investigating for its own sake. 
In this Section, we first give some illustrations of the expressivity of this subclass,
and we then study the question of what it is to be \emph{not} BIT-decomposable. 
By lack of space, the discussion in this Section remains rather informal with very rough proof sketches,
and more formal developments are left for a future work.
Also, the first two subsections are extended results based on previously published papers
(respectively~\cite{parco:csr2017} and \cite{parco-arch-aofa18}).

\subsection{From tree Posets to fork-join parallelism}
\label{sec:fj}

If the control-graph of a process is decomposed with only the B(ottom) rule
(or equivalently the T(op) rule), then it is rather easy to show that its
shape is that of a \emph{tree}.
These are processes that cannot do much beyond forking sub-processes.
For example, based on our language of barrier synchronization
it is very easy to encode e.g. the (rooted) binary trees:
$$T ::= 0 \mid \alpha.(T \parallel T) \quad \text{or e.g.} \quad T ::= 0 \mid {\nu}B~\left(\alpha.\langle B \rangle 0 \parallel \langle B \rangle T \parallel \langle B \rangle T \right)$$
The good news is that the combinatorics on trees
is well-studied. In the paper~\cite{BGP13} we provide a thorough study of such processes,
and in particular we describe very efficient counting and uniform random generation algorithms.
Of course, this is not a very interesting sub-class in terms of concurrency.

\begin{table}
\begin{center}
\begin{tabular}{cm{8pt}cm{8pt}cm{8pt}cm{8pt}c}
$\infer{\sigma \entails{FJ} 0}{}$ && $\infer{\sigma \entails{FJ} \alpha.P}{\sigma \entails{FJ} P}$ && $\infer{\sigma \entails{FJ} P \parallel Q}{\sigma \entails{FJ} P & \sigma \entails{FJ} Q}$
&& $\infer{\sigma \entails{FJ} \nu(B)~P}{B{::}\sigma \entails{FJ} P}$ && $\infer{B{::}\sigma \entails{FJ} \langle B \rangle.P}{\sigma \entails{FJ} P}$
\end{tabular}
\end{center}

\caption{\label{tab:fjproof} A proof system for fork-join processes.}
\end{table}

Thankfully, many results on trees generalize rather straightforwardly to \emph{fork-join parallelism},
a sub-class we characterize inductively in Table~\ref{tab:fjproof}.
Informally, this proof system imposes that processes use their synchronization barriers according
to a \emph{stack discipline}. When synchronizing, only the last created barrier is available,
which exactly corresponds to the traditional notion of a \emph{join} in concurrency.
Combinatorially, there is a correspondence between these processes and the class
of \emph{series-parallel Posets}. In the decomposition both the (B) and the (I) rule are needed,
but following a tree-structured strategy. Most (if not all) the interesting questions about such
partial orders can be answered in (low) polynomial time.

\begin{theorem}[cf.~\cite{parco:csr2017}]
For a fork join process of size $n$ the counting problem is of time complexity $O(n)$ and
we developed a bit-optimal uniform random sampler with time complexity $O(n\sqrt{n})$ on average.
\end{theorem}
%

\subsection{Asynchronism with promises}
\label{sec:promise}

We now discuss another interesting sub-class of processes
that can also be characterized inductively on the syntax of our process calculus, but
this time using the three BIT-decomposition rules (in a controlled manner).
The strict stack discipline of fork-join processes imposes a form of \emph{synchronous} behavior:
all the forked processes must terminate before a join may be performed.
To support a limited form of \emph{asynchronism}, a basic principle is to introduce \emph{promise} processes.

\begin{table}
\begin{tabular}{cm{8pt}cm{8pt}cm{8pt}c}
$\infer{\emptyset \entails{ctrl} 0}{}$ && $\infer{\pi \entails{ctrl} \alpha.P}{\pi \entails{ctrl} P}$  && $\infer{\pi \cup\{B\} \entails{ctrl} \langle B \rangle.P}{\pi \entails{ctrl} P}$  
&& $\infer{\pi \entails{ctrl} \nu(B)\left ( P \parallel Q\right )}{B\notin \pi & \pi \cup \{B\} \entails{ctrl} P & Q \uparrow_B}$  
\end{tabular}

\vspace{8pt}
with $Q \uparrow_B$ iff $Q\equiv \alpha.R$ and $R\uparrow_B$ or $Q\equiv \langle B \rangle.0$
\vspace{8pt}
\caption{\label{tab:pmproof} A proof system for promises.}
\end{table}

In Table~\ref{tab:pmproof} we define a simple inductive process structure composed as follows.
A \emph{main control thread} can perform atomic actions (at any time), and
also fork a sub-process of the form $\nu(B)\left ( P \parallel Q\right )$ but with a strong restriction:
\begin{itemize}
	\item a single barrier $B$ is created for the sub-processes to interact.
	\item the left sub-process $P$ must be the continuation of the main control thread,
	\item the right sub-process $Q$ must be a promise, which can only perform a sequence of atomic actions and ultimately synchronize with the control thread.
\end{itemize}

We are currently investigating this class as a whole, but we already obtained interesting results for the
\emph{arch-processes} in~\cite{parco-arch-aofa18}. An arch-process follows the constraint of
Table~\ref{tab:pmproof} but adds further restrictions.
The main control thread can still spawn an arbitrary number of promises, however there
must be two separate phases for the synchronization. After the first promise synchronizes, the main control thread cannot spawn any new promise. In~\cite{parco-arch-aofa18} a supplementary constraint is added (for the sake of algorithmic efficiency): 
each promise must perform exactly one atomic action, and the control thread can
 only perform actions when all the promises are running. In this paper, we remove this rather artificial constraint considering a larger, and more
useful process sub-class.

\begin{figure}[ht]
\begin{tabular}{c|c}
\begin{tikzpicture}[xscale=0.35, yscale=0.35] 
  
  \node (a1) at (0,5) {$\;\;\;\bullet^{a_1}_{\phantom{r}}$};
  \draw[->,>=latex,semithick] (a1) arc[radius=5, start angle=90, end angle=100] node (a11) {$\prescript{a_{1,1}}{}\bullet_{\;\;\;\;\;\;\;}$};
  \draw[dotted,->,>=latex,semithick] (a11) arc[radius=5, start angle=100, end angle=120] node (a1r) {$\prescript{a_{1,r_1}}{}\bullet_{\;\;\;\;\;\;\;}$};
  \draw[->,>=latex,semithick] (a1r) arc[radius=5, start angle=120, end angle=130] node (a2) {$\prescript{a_2 \;}{}\bullet_{\;\;\;\;}$};
  \draw[dotted,->,>=latex,semithick] (a2) arc[radius=5, start angle=130, end angle=160] node (al) {$\prescript{a_k}{}\bullet_{\;\;\;\;}$};
  \draw[->,>=latex,semithick] (al) arc[radius=5, start angle=160, end angle=170] node (ak1) {$\prescript{a_{k,1}}{}\bullet_{\;\;\;\;\;\;}$};
  \draw[dotted,->,>=latex,semithick] (ak1) arc[radius=5, start angle=170, end angle=190] node (akr) {$\prescript{a_{k,r_k}}{}\bullet_{\;\;\;\;\;\;\;\;}$};
  \draw[->,>=latex,semithick] (akr) arc[radius=5, start angle=190, end angle=200] node (c1) {$\prescript{c_1 \;}{}\bullet_{\;\;\;\;}$};
  \draw[->,>=latex,semithick] (c1) arc[radius=5, start angle=200, end angle=210] node (c11) {$\prescript{c_{1,1}} {}\bullet_{\;\;\;\;\;\;}$};
 \draw[dotted,->,>=latex,semithick] (c11) arc[radius=5, start angle=210, end angle=230] node (c1t) {$\prescript{c_{1,t_1}\;}{}\bullet_{\;\;\;\;\;\;}$};
  \draw[->,>=latex,semithick] (c1t) arc[radius=5, start angle=230, end angle=240] node (c2) {$\prescript{\phantom{f}}{c_2\;}\bullet_{\;\;\;\;\;}$};
  \draw[dotted,->,>=latex,semithick] (c2) arc[radius=5, start angle=240, end angle=270] node (cl) {$\;\;\;\;\bullet^{c_k}$};
  
  \draw[-<,>=latex,semithick] (c1) arc[radius=9.7, start angle=-60, end angle=-34] node (b11) {};
  \draw[<-,>=latex,semithick] (c1) arc[radius=9.7, start angle=-60, end angle=-37] node (b1) {$\;\;\;\;\;\;\;\bullet_{b_{1,s_1}}$};
  \draw[dotted,semithick] (c1) arc[radius=9.7, start angle=-60, end angle=-25] node (b12) {$\;\;\;\;\;\;\;\bullet_{b_{1,1}}\;$};
  \draw[dotted,-<,>=latex,semithick] (c1) arc[radius=9.7, start angle=-60, end angle=-22] node (b121) {};
  \draw[semithick] (b12) arc[radius=9.7, start angle=-25, end angle=-10] node (a1) {};
  
  \draw[-<,>=latex,semithick] (c2) arc[radius=9.7, start angle=-20, end angle=-7] node (b21) {};
  \draw[<-,>=latex,semithick] (c2) arc[radius=9.7, start angle=-20, end angle=-10] node (b2) {$\;\;\;\;\;\;\;\bullet_{b_{2,s_2}}$};
  \draw[dotted,semithick] (c2) arc[radius=9.7, start angle=-20, end angle=20] node (b22) {$\;\;\;\;\;\;\bullet^{b_{2,1}}$};
  \draw[dotted,-<,>=latex,semithick] (c2) arc[radius=9.7, start angle=-20, end angle=21] node (b221) {};
  \draw[semithick] (b22) arc[radius=9.7, start angle=17, end angle=27] node (a2) {};  
  
  \draw[-<,>=latex,semithick] (cl) arc[radius=9.7, start angle=10, end angle=33] node (bl1) {};
  \draw[<-,>=latex,semithick] (cl) arc[radius=9.7, start angle=10, end angle=30] node (bl) {$\;\;\;\;\;\;\bullet_{b_{k,s_k}}$};
  \draw[dotted,semithick] (cl) arc[radius=9.7, start angle=10, end angle=50] node (bl2) {$\;\;\;\;\;\;\bullet^{b_{k,1}}$};
  \draw[dotted,-<,>=latex,semithick] (cl) arc[radius=9.7, start angle=10, end angle=51] node (bl21) {};
  \draw[semithick] (bl2) arc[radius=9.7, start angle=47, end angle=57] node (al) {};  

\end{tikzpicture}

\phantom{XXX}

&

\phantom{XXX}

        \begin{tikzpicture}[xscale=0.35, yscale=0.35] 
        \node (a1) at (0,5) {$\;\;\;\bullet^{a_1}_{\phantom{r}}$};
        \draw[->,>=latex,semithick] (a1) arc[radius=5, start angle=90, end angle=100] node (a11) {$\prescript{a_{1,1}}{}\bullet_{\;\;\;\;\;\;\;}$};
        \draw[dotted,->,>=latex,semithick] (a11) arc[radius=5, start angle=100, end angle=120] node (a1r) {$\prescript{a_{1,r_1}}{}\bullet_{\;\;\;\;\;\;\;}$};
        \draw[->,>=latex,semithick] (a1r) arc[radius=5, start angle=120, end angle=130] node (a2) {$\prescript{a_2 \;}{}\bullet_{\;\;\;\;}$};
        \draw[dotted,->,>=latex,semithick] (a2) arc[radius=5, start angle=130, end angle=160] node (al) {$\prescript{a_k}{}\bullet_{\;\;\;\;}$};
        \draw[->,>=latex,semithick] (al) arc[radius=5, start angle=160, end angle=170] node (ak1) {$\prescript{a_{k,1}}{}\bullet_{\;\;\;\;\;\;}$};
        \draw[dotted,->,>=latex,semithick] (ak1) arc[radius=5, start angle=170, end angle=190] node (akr) {$\prescript{a_{k,r_k}}{}\bullet_{\;\;\;\;\;\;\;\;}$};
        \draw[->,>=latex,semithick] (akr) arc[radius=5, start angle=190, end angle=200] node (c1) {$\prescript{c_1 \;}{}\bullet_{\;\;\;\;}$};
        \draw[->,>=latex,semithick] (c1) arc[radius=5, start angle=200, end angle=210] node (c11) {$\prescript{c_{1,1}} {}\bullet_{\;\;\;\;\;\;}$};
        \draw[dotted,->,>=latex,semithick] (c11) arc[radius=5, start angle=210, end angle=230] node (c1t) {$\prescript{c_{1,t_1}\;}{}\bullet_{\;\;\;\;\;\;}$};
        \draw[->,>=latex,semithick] (c1t) arc[radius=5, start angle=230, end angle=240] node (c2) {$\prescript{\phantom{f}}{c_2\;}\bullet_{\;\;\;\;\;}$};
        \draw[dotted,->,>=latex,semithick] (c2) arc[radius=5, start angle=240, end angle=270] node (cl) {$\;\;\;\;\bullet^{c_k}$};
        
        \draw[semithick] (c1) arc[radius=9.7, start angle=-60, end angle=-25] node (b12) {$\;\;\;\;\;\;\;\bullet_{b_{1,1}}\;$};
        \draw[dotted,-<,>=latex,semithick] (c1) arc[radius=9.7, start angle=-60, end angle=-22] node (b121) {};
        \draw[semithick] (b12) arc[radius=9.7, start angle=-25, end angle=-10] node (a1) {};

\node at (-2.6,1) {$\mathcal{P}$};

        \end{tikzpicture}
        \begin{tikzpicture}[xscale=0.35, yscale=0.35] 
        \node (a1) at (0,5) {$\;\;\;\bullet^{a_1}_{\phantom{r}}$};
        \draw[->,>=latex,semithick] (a1) arc[radius=5, start angle=90, end angle=100] node (a11) {$\prescript{a_{1,1}}{}\bullet_{\;\;\;\;\;\;\;}$};
        \draw[dotted,->,>=latex,semithick] (a11) arc[radius=5, start angle=100, end angle=120] node (a1r) {$\prescript{a_{1,r_1}}{}\bullet_{\;\;\;\;\;\;\;}$};
        \draw[->,>=latex,semithick] (a1r) arc[radius=5, start angle=120, end angle=130] node (a2) {$\prescript{a_2 \;}{}\bullet_{\;\;\;\;}$};
        \draw[dotted,->,>=latex,semithick] (a2) arc[radius=5, start angle=130, end angle=160] node (al) {$\prescript{a_k}{}\bullet_{\;\;\;\;}$};
        \draw[->,>=latex,semithick] (al) arc[radius=5, start angle=160, end angle=170] node (ak1) {$\prescript{a_{k,1}}{}\bullet_{\;\;\;\;\;\;}$};
        \draw[dotted,->,>=latex,semithick] (ak1) arc[radius=5, start angle=170, end angle=190] node (akr) {$\prescript{a_{k,r_k}}{}\bullet_{\;\;\;\;\;\;\;\;}$};
        \draw[->,>=latex,semithick] (akr) arc[radius=5, start angle=190, end angle=200] node (c1) {$\prescript{c_1 \;}{}\bullet_{\;\;\;\;}$};
        \draw[->,>=latex,semithick] (c1) arc[radius=5, start angle=200, end angle=210] node (c11) {$\prescript{c_{1,1}} {}\bullet_{\;\;\;\;\;\;}$};
        \draw[dotted,->,>=latex,semithick] (c11) arc[radius=5, start angle=210, end angle=230] node (c1t) {$\prescript{c_{1,t_1}\;}{}\bullet_{\;\;\;\;\;\;}$};
        \draw[->,>=latex,semithick] (c1t) arc[radius=5, start angle=230, end angle=240] node (c2) {$\prescript{\phantom{f}}{c_2\;}\bullet_{\;\;\;\;\;}$};
        \draw[dotted,->,>=latex,semithick] (c2) arc[radius=5, start angle=240, end angle=270] node (cl) {$\;\;\;\;\bullet^{c_k}$};

        \draw[semithick,blue] (cl) arc[radius=9.7, start angle=10, end angle=50] node (b12) {$\;\;\;\;\;\;\bullet^{b_{1,1}}$};
        \draw[-<,>=latex,semithick,blue] (cl) arc[radius=9.7, start angle=10, end angle=51] node (b121) {};
        \draw[semithick,blue] (cl) arc[radius=9.7, start angle=10, end angle=60] node (al) {};
        
        \draw[semithick,red] (cl) arc[radius=13, start angle=-22, end angle=13] node (b12) {$\prescript{b_{1,1}}{}\bullet\;\;\;\;\;$};
        \draw[dotted,-<,>=latex,semithick,red] (cl) arc[radius=13, start angle=-22, end angle=15] node (b121) {};
        \draw[semithick,red] (cl) arc[radius=13, start angle=-22, end angle=23] node (a1) {};
        
        \draw[semithick,green] (c1) arc[radius=2.2, start angle=-50, end angle=-12] node (b12) {$\;\;\;\;\;\;\bullet_{b_{1,1}}\;$};
        \draw[-<,>=latex,semithick,green] (c1) arc[radius=2.2, start angle=-50, end angle=10] node (b121) {};
        \draw[semithick,green] (c1) arc[radius=2.2, start angle=-50, end angle=50] node (a1) {};

\node[red] at (1.6,1) {$\mathcal{A}$};
\node[blue] at (-0.6,-1.4) {$\mathcal{B}$};
\node[green] at (-3.7, -1.4) {$\mathcal{C}$};
        \end{tikzpicture}

\end{tabular}

\caption{\label{fig:arch} The structure of an arch-process (left) and the inclusion-exclusion counting principle (right).}
\end{figure}
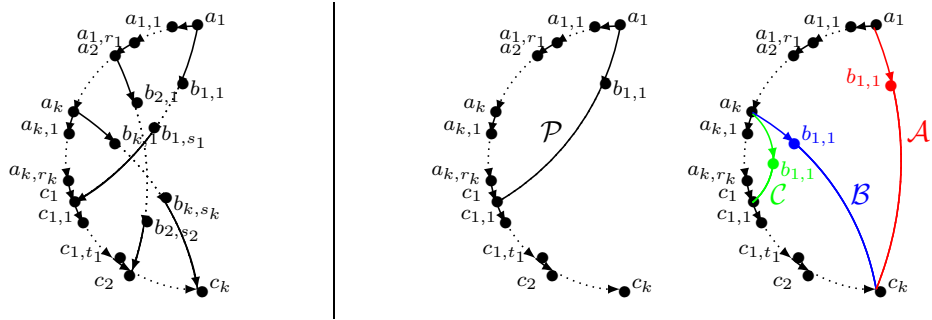

In Fig.~\ref{fig:arch} (left) is represented the general structure of a generalized arch-process.
The $a_i$'s actions are the promise forks, and the synchronization points are the $c_j$'s. 
The constraint is thus that all the $a_i$'s occur before the $c_j$'s. 

\begin{theorem}\label{theo:nb_exc_promise}
The number of executions of an arch-process can be calculated in $O(n^2)$ arithmetic operations,
using a dynamic programming algorithm based on memoization.    
\end{theorem}

\begin{proof}[idea]
A complete proof is provided in~\cite{parco-arch-aofa18} for ``simple'' arch-processes,
and the generalization is detailed in the companion document.
We only describe the \emph{inclusion-exclusion} principle on which
our counting algorithm is based. Fig.~\ref{fig:arch} (right) describes this principles
(we omit the representation of the other promises to obtain a clear picture of our approach).
Our objective is to count the number of execution contributed
by a single promise with atomic action $b_{1,1}$.
If we denote by $\ell_\mathcal{P}$ this contribution, 
we reformulate it as a combination
$\ell_\mathcal{P} = \ell_{\color{red}{\mathcal{A}}} - \ell_{\color{blue}{\mathcal{B}}} + \ell_{\color{green}{\mathcal{C}}}$
as depicted on the rightmost part of Fig.~\ref{fig:arch}.
First, we take the ``virtual'' promise $\mathcal{A}$ going from the starting point $a_1$ of $\ell_\mathcal{P}$
until the end point $c_{k}$ of the main thread.
Of course there are two many possibilities if we only keep $\textcolor{red}{\mathcal{A}}$.
An over-approximation of what it is to remove is the promise $\textcolor{blue}{\mathcal{B}}$
going from the start of the last promise (at point $a_k$) until the end.
But this time we removed too many possibilities, which corresponds to promise $\textcolor{green}{\mathcal{C}}$.
The latter is thus reinserted in the count. Each of these three ``virtual'' promises
have a simpler counting procedure. To guarantee the quadratic worst-time complexity (in the number of arithmetic operations),
we have to memoize the intermediate results. We refer to the companion document for further details. \hfill \qed
\end{proof}

From this counting procedure we developed a uniform random sampler
following the principles of the \emph{recursive method}, as described in~\cite{FZC94}.

\begin{theorem}\label{theo:random_gen_promise}
    Let $\mathcal{P}$ be a promise-process of size $n$ with $k\geq n$ promises.
    A random sampler of $O(n^4)$ time-complexity (in the number of arithmetic operations)
    builds uniform executions. 
\end{theorem}

The algorithm and the complete proof are detailed in the companion document.
   One notable aspect is that in order to get rid of the forbidden case
   of executions associated to the ``virtual'' promise $\textcolor{blue}{\mathcal{B}}$
   we cannot only do rejection (because the induced complexity would be exponential).
In the generalization of arch-processes, we proceed by case analysis: for each possibility for the insertion of $b_{1,1}$
   in the main control thread we compute the relative probability for the associated process $\mathcal{P}$. 
This explains the increase of complexity (from $O(n^2)$ to $O(n^4)$) if compared to \cite{parco-arch-aofa18}.

\subsection{BIT-free processes}
\label{sec:cf}

The class of BIT-decomposable processes is rather large, and we in fact only uncovered two
interesting sub-classes that can be easily captured inductively on the process syntax.
The relatively non-trivial process $Sys$ of Fig.~\ref{fig:example} is also interestingly BIT-decomposable.
We now adopt the complementary view of trying to understand the combinatorial structure of
 a so called ``BIT-free'' process, which is \emph{not} decomposable using only the (B), (I) and (T) rules.

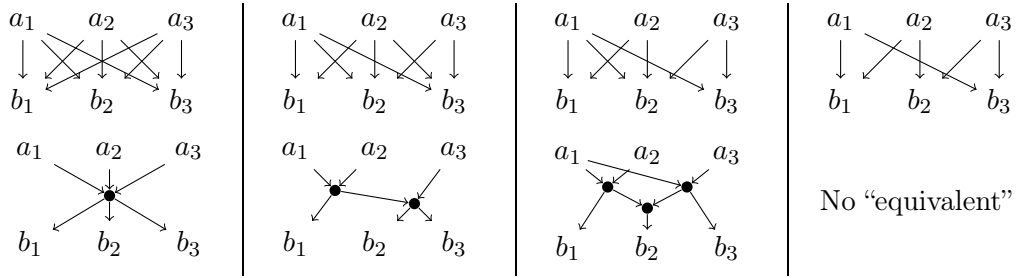
\begin{figure}[ht]
\begin{tabular}{c|c|c|c}
\begin{tikzpicture}[node distance=30pt]
\node (a1) {$a_1$};
\node[right of=a1] (a2) {$a_2$};
\node[right of=a2] (a3) {$a_3$};
\node[below of=a1, node distance=30pt] (b1) {$b_1$};
\node[right of=b1] (b2) {$b_2$};
\node[right of=b2] (b3) {$b_3$};
\draw[->] (a1) -- (b1);
\draw[->] (a1) -- (b2);
\draw[->] (a1) -- (b3);
\draw[->] (a2) -- (b1);
\draw[->] (a2) -- (b2);
\draw[->] (a2) -- (b3);
\draw[->] (a3) -- (b2);
\draw[->] (a3) -- (b3);
\draw[->] (a3) -- (b1);
\end{tikzpicture} \hspace{2pt} & \hspace{2pt} \begin{tikzpicture}[node distance=30pt]
\node (a1) {$a_1$};
\node[right of=a1] (a2) {$a_2$};
\node[right of=a2] (a3) {$a_3$};
\node[below of=a1, node distance=30pt] (b1) {$b_1$};
\node[right of=b1] (b2) {$b_2$};
\node[right of=b2] (b3) {$b_3$};
\draw[->] (a1) -- (b1);
\draw[->] (a1) -- (b2);
\draw[->] (a1) -- (b3);
\draw[->] (a2) -- (b1);
\draw[->] (a2) -- (b2);
\draw[->] (a2) -- (b3);
\draw[->] (a3) -- (b2);
\draw[->] (a3) -- (b3);
\end{tikzpicture}  \hspace{2pt} & \hspace{2pt} \begin{tikzpicture}[node distance=30pt]
\node (a1) {$a_1$};
\node[right of=a1] (a2) {$a_2$};
\node[right of=a2] (a3) {$a_3$};
\node[below of=a1, node distance=30pt] (b1) {$b_1$};
\node[right of=b1] (b2) {$b_2$};
\node[right of=b2] (b3) {$b_3$};
\draw[->] (a1) -- (b1);
\draw[->] (a1) -- (b2);
\draw[->] (a1) -- (b3);
\draw[->] (a2) -- (b1);
\draw[->] (a2) -- (b2);
\draw[->] (a3) -- (b2);
\draw[->] (a3) -- (b3);
\end{tikzpicture} \hspace{2pt} & \hspace{2pt}  \begin{tikzpicture}[node distance=30pt]
\node (a1) {$a_1$};
\node[right of=a1] (a2) {$a_2$};
\node[right of=a2] (a3) {$a_3$};
\node[below of=a1, node distance=30pt] (b1) {$b_1$};
\node[right of=b1] (b2) {$b_2$};
\node[right of=b2] (b3) {$b_3$};
\draw[->] (a1) -- (b1);
\draw[->] (a1) -- (b3);
\draw[->] (a2) -- (b1);
\draw[->] (a2) -- (b2);
\draw[->] (a3) -- (b2);
\draw[->] (a3) -- (b3);
\end{tikzpicture} \\
\begin{tikzpicture}[node distance=30pt]
\node (a1) {$a_1$};
\node[right of=a1] (a2) {$a_2$};
\node[right of=a2] (a3) {$a_3$};
\node[below of=a1, node distance=35pt] (b1) {$b_1$};
\node[right of=b1] (b2) {$b_2$};
\node[right of=b2] (b3) {$b_3$};
\node[above of=b2, node distance=18pt, shape=circle, fill, inner sep=1.5pt] (c1) {};
\draw[->] (a1) -- (c1);
\draw[->] (a2) -- (c1);
\draw[->] (a3) -- (c1);
\draw[->] (c1) -- (b1);
\draw[->] (c1) -- (b2);
\draw[->] (c1) -- (b3);
\end{tikzpicture} & \begin{tikzpicture}[node distance=30pt]
\node (a1) {$a_1$};
\node[right of=a1] (a2) {$a_2$};
\node[right of=a2] (a3) {$a_3$};
\node[below of=a1, node distance=35pt] (b1) {$b_1$};
\node[right of=b1] (b2) {$b_2$};
\node[right of=b2] (b3) {$b_3$};
\node[below of=a2, node distance=20pt] (cc1) {};
\node[right of=cc1, node distance=15pt, shape=circle, fill, inner sep=1.5pt] (c1) {};
\node[below of=a2, node distance=15pt] (cc2) {};
\node[left of=cc2, node distance=15pt, shape=circle, fill, inner sep=1.5pt] (c2) {};

\draw[->] (a1) -- (c2);
\draw[->] (a2) -- (c2);
\draw[->] (a3) -- (c1);
\draw[->] (c2) -- (c1);
\draw[->] (c2) -- (b1);
\draw[->] (c1) -- (b2);
\draw[->] (c1) -- (b3);
\end{tikzpicture} &  \begin{tikzpicture}[node distance=30pt]
\node[inner sep=2pt] (a1) {$a_1$};
\node[right of=a1,inner sep=2pt] (a2) {$a_2$};
\node[right of=a2,inner sep=2pt] (a3) {$a_3$};
\node[below of=a1, node distance=35pt,inner sep=2pt] (b1) {$b_1$};
\node[right of=b1,inner sep=2pt] (b2) {$b_2$};
\node[right of=b2,inner sep=2pt] (b3) {$b_3$};
\node[below of=a2, node distance=25pt] (cc1) {};
\node[below of=a2, node distance=12pt] (cc2) {};
\node[left of=cc2, node distance=15pt, shape=circle, fill, inner sep=1.5pt] (c2) {};
\node[right of=cc2, node distance=15pt, shape=circle, fill, inner sep=1.5pt] (c3) {};
\node[above of=b2, node distance=15pt, shape=circle, fill, inner sep=1.5pt] (c1) {};

\draw[->] (a1) -- (c2);
\draw[->] (a2) -- (c2);
\draw[->] (a1) -- (c3);
\draw[->] (a3) -- (c3);
\draw[->] (c2) -- (c1);
\draw[->] (c3) -- (c1);
\draw[->] (c2) -- (b1);
\draw[->] (c1) -- (b2);
\draw[->] (c3) -- (b3);
\end{tikzpicture} &
\begin{tikzpicture}
 \node (a) {};
 \node[below of=a, node distance=20pt] (n) {No ``equivalent''};
 \node[below of=n, node distance=20pt] (b) {};
 \end{tikzpicture}
\end{tabular}
\caption{\label{fig:crowns}Typical BIT-free substructures, and their BIT ``equivalent'' (when possible).}
\end{figure}

The BIT-free condition implies the occurrence of structures similar to the ones depicted on Fig.~\ref{fig:crowns}.
These structures are composed of a set of ``bottom'' processes (the $b_i$'s) waiting for ``top'' processes
 (the $a_j$') according to some synchronization pattern. We represent the whole possibilities
  of size 3 (up-to order-isomorphism) in the upper-part of the figure.
  The upper-left process is a complete (directed) bipartite graph, which can in fact be ``translated''
  to a BIT-decomposable process as seen on the lower-part of the figure. This requires the
  introduction of a single ``synchronization point'' between the two process groups. This
  transformation preserves the number of executions and is Poset-wise equivalent. At each step
   ``to the right'' of Fig.~\ref{fig:crowns}, we remove a directed edge. In the second and third processes (in the middle), we also
   have an equivalent with respectively two and three synchronization points. In these cases, the number of linear extensions is
   not preserved but the ``nature'' of the order is respected: the interleavings of the initial atomic actions are the same.
   The only non-transformable structure,
    let's say the one ``truly'' BIT-free is the rightmost process. Even if we introduce 
    synchronization points (we need at least three of them), the structure would not become BIT-decomposable. In terms of order
    theory such a structure is called a \emph{Crown} poset. In~\cite{posets:height2} it is shown
     that the counting problem is already $\sharp$-P complete for partial orders of height 2, hence directed bipartite digraphs similar to the structures of Fig.~\ref{fig:crowns}. One might wonder if this is still the case when these
     structures cannot occur, especially in the case of BIT-decomposable processes.
      This is for us a very interesting (and open) problem.

\section{Experimental study}
\label{sec:bench}

\begin{table}[ht]
\begin{tabular}{c|c|c|c|c}
\textbf{Algorithm} & \textbf{Class} & \textbf{Count.} & \textbf{Unif. Rand. Gen.} & \textbf{Reference} \\
\hline
\textsc{FJ} & Fork-join & $O(n)$ & $O(n\cdot \sqrt{n})$ on average & \cite{parco:csr2017} \\
\textsc{Arch} & Arch-processes & $O(n^2)$ & $O(n^4)$ worst case & \cite{parco-arch-aofa18}/Theorem~\ref{theo:random_gen_promise} \\
\textsc{bit} & BIT-decomposable & ? & ? & Theorem~\ref{thm:bitc:integral} \\
\textsc{cftp}\footnotemark{} & All processes & --  & $O(n^3\cdot log~n)$ expected & \cite{huber:dm06} \\
\end{tabular}
\vspace{8pt}
\caption{\label{tab:algos} Summary of counting and uniform random sampling algorithms (time complexity figures with $n$: number of atomic actions).}
\end{table}

\footnotetext{The \textsc{cftp} algorithm is the only one we did not design, but only implement. Its complexity is $O(n^3\cdot log~n)$ (randomized) expected time.}

In this section, we put into use the various algorihms for counting and generating process executions uniformly at random. Table~\ref{tab:algos} summarizes these algorithms and the associated worst-case time complexity bounds (when known). We implemented all the algorithms in Python 3, and we did not optimize for efficiency, hence the numbers we obtain only give a rough idea of their performances. For the sake of reproducibility, the whole experimental setting is available in the companion repository, with explanations about the required dependencies and usage. The computer we used to perform the benchmark is a standard laptop PC with an I7-8550U CPU, 8Gb RAM running Manjaro Linux. As an initial experiment, the example of Fig.~\ref{fig:example} is BIT-decomposable, so we can apply the \textsc{bit} and \textsc{cftp} algorithms. The counting (of its $1975974$ possible executions) takes about 0.3s and it takes about 9 millisecond to uniformly generate an execution with the \textsc{bit} sampler, and about 0.2s with \textsc{cftp}. For ``small'' state spaces, we observe that \textsc{bit} is always faster than \textsc{cftp}.

\begin{table}[htb]
\begin{center}
\begin{tabular}{c|c|rr|rr|r}
  \textbf{FJ size} & \textbf{$\sharp\LE$} & \textsc{FJ} \textbf{gen} & \textbf{(count)} & \textsc{bit} \textbf{gen} & \textbf{(count)} & \textsc{cftp} \textbf{gen} \\ 
\hline
10 & 19 & 0.00001 s & (0.0002 s) & 0.0006 s & (0.03 s) & 0.04 s\\
30 & $10^{9}$ & 0.00002 s & (0.0002 s) & 0.02 s & (0.03 s) & 1.8 s \\
40 & $6\cdot 10^{6}$ & 0.00004 s & (0.0003 s) & 3.5 s & (5.2 s) & 5.6 s \\ 
63 & $4\cdot 10^{29}$ & 0.0005 s & (0.03 s) & Mem. crash & (Crash) & 55 s \\
217028 & $2\cdot 10^{292431}$ &  8.11 s & (3.34 s) & Mem. crash & (Crash) & Timeout \\
\end{tabular}

\vspace{2ex}

\begin{tabular}{c|c|rr|rr|r}
\textbf{Arch size} & \textbf{$\sharp\LE$} & \textsc{Arch} \textbf{gen} & \textbf{(count)} & \textsc{bit} \textbf{gen} & \textbf{(count)} & \textsc{cftp} \textbf{gen} \\ 
\hline
10:2 & 43 & 0.00002 s & (0.00004 s) & 0.002 s & (0.000006 s) & 0.04 s\\
30:2 & $9.8\cdot 10^{8}$ & 0.003 s & (0.0009 s) & 0.000007 s & (0.0004 s) & 1.5 s\\
30:4 & $6.9\cdot 10^{10}$ & 0.001 s & (0.005 s) & 0.000007 s & (0.004 s) & 2.5 s\\
100:2 & $1.3\cdot 10^{32}$ & 0.75 s & (0.16 s) & Mem. crash & (Crash) & $\;^6\;$ 5.6 s \\
100:32 & $1\cdot 10^{53}$ & 2.7 s & (0.17 s) & Mem. crash & (Crash) & $\;^6\;$ 5.9 s \\
200:66 & $10^{130}$ &  54 s & (31 s) & Mem. crash & (Crash) & Timeout \\
\end{tabular}
\end{center}

\vspace{1ex}

\caption{\label{tab:bench} Benchmark results for BIT-decomposable classes: FJ and Arch.}
\end{table}

\footnotetext{For arch-processes of size $100$ with $2$ arches or $32$,
 the \textsc{cftp} algorithm timeouts (30s) for almost all of the input graphs.}

For a more thorough comparison of the various algorithms, we generated random processes
(uniformly at random among all processes of the same size) in the classes of fork-join (FJ) and arch-processes as discussed in Section~\ref{sec:subclasses}, using our own Arbogen tool\footnote{Arbogen is uniform random generation for context-free grammar structures: cf.~\url{https://github.com/fredokun/arbogen}.}
or an ad hoc algorithm for arch-processes (presented in the companion repository). 
For the fork-join structures, the size is simply the number of atomic actions in the process.
It is not a surprise that the dedicated algorithms we developed in~\cite{parco:csr2017} outperforms
the other algorithms by a large margin. In a few second it can handle extremely large state spaces,
which is due to the large ``branching factor'' of the process ``forks''.
The arch-processes represent a more complex structure, thus the numbers are less ``impressive'' than in the FJ case.
To generate the arch-processes (uniformly at random), we used the number of atomic actions
as well as the number of spawned promises as main parameters.
Hence an arch of size `$n$:$k$' has $n$ atomic actions and $k$ spawned promises.
Our dedicated algorithm for arch-process is also rather effective, considering the state-space sizes it can handle.
In less than a minute it can generate an execution path uniformly at random for a process of size 200 with 66 spawned promises,
the state-space is in the order of $10^{130}$.
Also, we observe that in all our tests the observable ``complexity'' is well below $O(n^4)$. The reason is that we perform the pre-computations
 (corresponding to the worst case) in a just-in-time (JIT) manner, and in practice we only actually need a small fractions of the computed values. However the random sampler is much more efficient with the separate precomputation. As an illustration, for arch-processes of size $100$ with $32$ arches,
the sampler becomes about 500 times faster. However the memory requirement for the precomputation grows very quickly, so that the JIT variant is clearly preferable.

In both the FJ and arch-process cases the current implementation of the \textsc{bit} algorithms is not entirely satisfying. One reason is that the strategy we employ for the BIT-decomposition is quite ``oblivious'' to the actual structure of the DAG. As an example, this strategy handles fork-joins far better than arch-processes. In comparison, the \textsc{cftp} algorithm is less sensitive to the structure, it performs quite uniformly on the whole benchmark. We are still confident that by handling the integral computation natively, the \textsc{bit} algorithms could handle much larger state-spaces. For now, they are only usable up-to a size of about 40 nodes (already corresponding to a rather large state space).

\section{Conclusion and future work}

The process calculus presented in this paper is quite limited in
terms of expressivity. In fact, as the paper makes clear it can
only be used to describe (intransitive) directed acyclic graphs!
However we still believe it is an interesting ``core synchronization calculus'',
providing the minimum set of features so that processes are
isomorphic to the whole combinatorial class of partially ordered
sets. Of course, to become of any practical use, the barrier synchronization
calculus should be complemented with e.g. non-deterministic choice (as we investigate
in~\cite{BGP13}). Moreover, the extension of our approach to iterative processes
remains full of largely open questions. 

Another interest of the proposed language is that it can be used to define process (hence poset) sub-classes
in an inductive way. We give two illustrations in the paper with the \emph{fork-join} processes and \emph{promises}.
This is complementary to definitions wrt. some combinatorial properties, such as the
``BIT-decomposable'' vs. ``BIT-free'' sub-classes. The class of arch-processes (that we study in~\cite{parco-arch-aofa18} and
generalize in the present paper) is also interesting: it is a combinatorially-defined sub-class of the inductively-defined
asynchronous processes with promises. We see as quite enlightening the meeting of these two distinct points of view.

Even for the ``simple'' barrier synchronizations, our study
is far from being finished because we are, in a way, also looking for ``negative'' results. 
The counting problem is hard, which is of course tightly related to the infamous ``combinatorial explosion'' phenomenon in concurrency. We in fact believe that the problem remains intractable for the class of BIT-decomposable processes, but this is still an open question that we intend to investigate furthermore.  By delimiting more precisely the ``hardness'' frontier, we hope to find more interesting sub-classes for which we can develop efficient counting and random sampling algorithms.

\vspace*{1ex}

\noindent \textbf{Acknowledgment}
We thank the anonymous reviewers as well as our ``shepard'' for helping us making the paper better and hopefully with less errors.

\bibliographystyle{plain}
\bibliography{parco}

\end{document}